\documentclass[11pt]{article}
\usepackage{mathtools}

\usepackage{amsfonts,amsmath,amssymb,amsthm}
\usepackage[margin=1in]{geometry}

\newcommand{\bE}{\ensuremath{\mathbf{E}}}

\newtheorem{theorem}{Theorem}[section]
\newtheorem{proposition}[theorem]{Proposition}
\newtheorem{lemma}[theorem]{Lemma}
\newtheorem{corollary}[theorem]{Corollary}
\newtheorem{definition}[theorem]{Definition}

\begin{document}

\title{Algorithmic and enumerative aspects of the Moser-Tardos distribution\footnote{A preliminary version of this paper has appeared in the \emph{Proc.\ ACM-SIAM Symposium on Discrete Algorithms}, 2016.}}
\author{
David G. Harris\thanks{Department of Computer Science, University of Maryland, 
College Park, MD 20742. 
Research supported in part by NSF Awards CNS-1010789 and CCF-1422569.
Email: \texttt{davidgharris29@gmail.com}} \\
\and
Aravind Srinivasan\thanks{Department of Computer Science and
Institute for Advanced Computer Studies, University of Maryland, 
College Park, MD 20742. 
Research supported in part by NSF Awards CNS-1010789 and CCF-1422569, and by a research award from Adobe, Inc. 
Email: \texttt{srin@cs.umd.edu}}
}

\date{}
\maketitle

\begin{abstract}
Moser \& Tardos have developed a powerful algorithmic approach (henceforth ``MT") to the Lov\'{a}sz Local Lemma (LLL); the basic operation done in MT and its variants is a search for ``bad" events in a current configuration. In the initial stage of MT, the variables are set independently. We examine the distributions on these variables which arise during intermediate stages of MT. We show that these configurations have a more or less ``random'' form, building further on the ``MT-distribution" concept of Haeupler et al.\ in understanding the (intermediate and) output distribution of MT. This has a variety of algorithmic applications; the most important is that bad events can be found relatively quickly, improving upon MT across the complexity spectrum: it makes some polynomial-time algorithms sub-linear (e.g., for Latin transversals, which are of basic combinatorial interest), gives lower-degree polynomial run-times in some settings, transforms certain super-polynomial-time algorithms into polynomial-time ones, and leads to Las Vegas algorithms for some coloring problems for which only Monte Carlo algorithms were known.  

We show that in certain conditions when the LLL condition is violated, a variant of the MT algorithm can still produce a distribution which avoids most of the bad events. We show in some cases this MT variant can run faster than the original MT algorithm itself, and develop the first-known criterion for the case of the asymmetric LLL. This can be used to find partial Latin transversals -- improving upon earlier bounds of Stein (1975) -- among other applications. We furthermore give applications in enumeration, showing that most applications (where we aim for all or most of the bad events to be avoided) have large solution sets. We do this by showing that the MT-distribution has large R\'{e}nyi entropy.

\end{abstract}

\smallbreak \noindent \textbf{Key words and phrases:} Lov\'{a}sz Local Lemma, Moser-Tardos algorithm, LLL-distribution, MT-distribution, graph coloring, satisfiability, Latin transversals, combinatorial enumeration.

\section{Introduction}
We consider a number of basic applications of the Lov\'{a}sz Local Lemma (LLL) in probabilistic combinatorics and graph theory \cite{alon-spencer}: these include Latin transversals, hypergraph $2$-coloring, various types of graph coloring, k-SAT, versions of these problems where we satisfy ``most" of the constraints (as in MAX-SAT), and enumerating (lower-bounding) the number of solutions to these problems. 
Recall that the LLL gives a powerful sufficient condition for avoiding all of a given set of bad events. 
 We study the seminal Moser-Tardos approach (henceforth ``MT") for algorithmic versions of the LLL \cite{moser-tardos}, presenting new analyses and branching processes to speed up the MT algorithm -- significantly in some cases (e.g., from exponential to polynomial, and from polynomial to sublinear); furthermore, we improve upon the known sufficient conditions for only a ``few" of the given bad events to occur. A fundamental idea behind our work is that the structures arising in the execution of MT are ``random-like", and that such average-case behavior can be used to good advantage. 

We refer to the distribution on the variables at the termination of the MT algorithm as the \emph{MT-distribution}. A key randomness property of this distribution has been demonstrated in \cite{hss}. We develop this further, showing that the intermediate structures arising in the execution of MT have some very useful ``random-like" properties, which can be exploited using additional ideas. 

In the MT setting, we have a set of \emph{variables} $X_1, \dots, X_n$. We have also a product probability distribution $\Omega$, which selects a integer value $j$ for each variable $X_i$ with probability $p_{i,j}$; the variables are drawn independently and $\sum_j p_{i,j} = 1$ for each $i$.
We have \emph{events}, which are Boolean functions of subsets of the variables. We say that $E \sim E'$ iff $E, E'$ overlap in some variable(s), i.e., if each of them involves some common $X_i$. (Note that we always have $E \sim E$.) There is a set of $m$ \emph{bad events $\mathcal B$} which we are trying to avoid.
In this setting, the MT algorithm is as follows:
\begin{enumerate}
\item[1.] Draw $X_1, \dots, X_n$ from $\Omega$.
\item[2.]  Repeat while there is some true bad event:
\begin{enumerate}
\item[2a.] Choose a currently-true bad event $B \in \mathcal B$ arbitrarily.
\item[2b.] Resample all the variables involved in $B$ from the restriction of $\Omega$ to just these variables. (We refer to this step as \emph{resampling the bad event $B$}).
\end{enumerate}
\end{enumerate}

For any event $E$ (whether in $\mathcal B$ or not), we let $N(E)$ denote the \emph{inclusive} neighborhood of $E$, viz.\ the set of all bad events $B \in \mathcal B$ such that $B \sim E$. This is ``inclusive" since $E \in N(E)$ for $E \in \mathcal B$.

When we are analyzing the MT algorithm, we let $T$ denote the termination time ($T = \infty$ if the algorithm runs forever). For $t = 0, \dots, T$ we let $X^t$ denote the configurations of the variables (the values of $X_1, \dots, X_n$) after $t$ resamplings; $X^0$ is the initial configuration (after step (1)). For $t = 1, \dots, T-1$, we let $B^t$ denote the bad-event which is resampled at time $t$.

In our analyses, there are two probability distributions at play. First, there is the distribution $\Omega$, to which the LLL applies and which the MT algorithm is (in a certain sense) trying to simulate. Second, there is the probability distribution which describes the execution of the MT algorithm; this second probability distribution is the one that is ``actually occurring.'' In order to ensure that this second probability distribution is well-defined, we assume that there is some fixed rule (possibly randomized) for choosing which bad-event to resample. We refer to probabilities of the first type as $P_{\Omega}$ and probabilities of the second type (which are the true probabilities of the events of interest) as simply $P$.

The key criterion for the convergence of the MT algorithm is the ``asymmetric LLL" \cite{spencer-ramsey}. We state a slightly stronger form of this criterion due to Pegden \cite{pegden}:
\begin{theorem}
\label{a1fapp-main-thm}
Suppose there is $\mu:\mathcal B \rightarrow \mathbf [0, \infty)$ such that for all $B \in \mathcal B$ we have
\begin{equation}
\label{a1eqn:pegden-asymm}
\mu(B) \geq P_{\Omega}(B) \times \sum_{\substack{I \subseteq N(B)\\\text{$I$ independent set under $\sim$}}}  \prod_{B' \in I} \mu(B').
\end{equation}
 Then the MT algorithm terminates with probability 1; the expected number of resamplings of any bad event $B \in \mathcal B$ is at most $\mu(B)$.   \footnote{Clearly, $\mu(B) \geq P_{\Omega}(B) \times \sum_{I \subseteq N(B)}  \prod_{B' \in I} \mu(B') = P_{\Omega}(B) \times \prod_{B' \in N(B)} (1 + \mu(B'))$ is a sufficient condition for (\ref{a1eqn:pegden-asymm}). Setting $x(B) = \mu(B)/(\mu(B) + 1)$ in this sufficient condition recovers the usual formulation of the asymmetric LLL.}
\end{theorem}

The ``Symmetric LLL'' is a special case of this, obtained by setting $\mu(B) = e \cdot P_{\Omega}(B)$: \footnote{In other formulations of the symmetric LLL, $N(B)$ is defined to be the \emph{exclusive} neighborhood (not counting $B$ itself), and hence the criterion becomes $e p (d+1) \leq 1$. The reader should bear in mind that in this paper, $N(B)$ non-standardly refers to the \emph{inclusive} neighborhood.}

\begin{theorem}
\label{a1fapp-main-thm2}
Suppose $P_{\Omega} (B) \leq p$ and $|N(B)| \leq d$ for all $B \in \mathcal B$, with $e p d \leq 1$. Then the MT algorithm terminates with probability 1, and the expected number of resamplings of any bad event is at most $ e p$.

\end{theorem}

The MT algorithm can give polynomial-time algorithms for nearly all applications of the Lov\'{a}sz Local Lemma. Yet, implemented directly, this algorithm can be fairly slow. The key bottleneck is that, in each step of the algorithm, one must search for currently-true bad events (or certify there are none). 
We show,
by understanding the MT-distribution and some of its relatives better, 
 that the configurations which arise during the execution of MT have a more or less ``random'' form, and that currently-true bad events can be found relatively quickly in expectation. 
Our main contributions are as follows.

\smallskip \noindent \textbf{(a) From super-polynomial to polynomial time, and from Monte Carlo to Las Vegas.} The MT algorithm, as described, may not run in poly$(n)$ time if the number of bad events is super-polynomial. This issue is addressed in \cite{hss}, where polynomial-time algorithms are developed for
many such cases. However, the framework of \cite{hss} has some important limitations. First, it typically requires satisfying the LLL criterion with an additional slack. This means that one typically obtains worse constructive bounds than the existential ones possible from the LLL. Second, this framework leads to Monte Carlo algorithms --- that is, the algorithm terminates and there is a high probability (but not certainty) of success. These problems are both present for the class of problems based on non-repetitive vertex colorings. In Section~\ref{a1fapp-sec:exp-to-poly}, we present improved algorithms for these problems; our algorithm leads to essentially the same parameters as the non-constructive LLL, and is Las Vegas. 

\smallskip \noindent \textbf{(b) Improved polynomial run-times.} We also significantly improve the run-times of certain combinatorial algorithms. In Section~\ref{a1sec:ramsey}, we give improved algorithms for Ramsey number lower bounds. In Section~\ref{a1fapp-sec:hyp2col}, we give improved algorithms for hypergraph 2-coloring, reducing a quadratic run-time to a quasi-linear run-time. In Section~\ref{a1fapp-sec:latin}, we give the first sub-linear algorithm for Latin transversals:
one that runs in time proportional to the square root of the input length.  Latin transversals and their ``partial transversal" variants are well-studied in combinatorics (see, e.g., \cite{bissacot,brualdi-ryser:book,erdos-spencer,lllperm,hatami-shor:latin,ryser:oberwolfach,stein:latin}), the latter of which we encounter in item (c) next. 

\smallskip \noindent \textbf{(c) Partially avoiding bad events.} In some cases, the LLL criterion is not satisfied, and one cannot necessarily avoid \emph{all} the bad events. However, one can still avoid most of the bad events. 
This issue was first examined in \cite{hss}, which extended the symmetric LLL to the case when $e p d = \alpha$, for $\alpha \in [1,e]$, and $d$ was large: they gave a randomized algorithm whose expected number of bad events at the end is $(1 + o(1)) \cdot mp \cdot (e \ln(\alpha)/\alpha) = (1 + o(1)) \cdot \frac{m \ln \alpha}{d} $, where the ``$o(1)$" term is a function of $d$ that tends to zero for large $d$. No such results were known for the general asymmetric LLL (Theorem~\ref{a1fapp-main-thm}) or symmetric LLL for small $d$. We develop the first ``few bad events" variant of Theorem~\ref{a1fapp-main-thm} in Theorem~\ref{a1falpha-resamp-thm1}, and also obtain an exact result for the symmetric LLL by removing the ``$o(1)$" term above (Corollary~\ref{a1fsym-cor-most}). 

These results apply to many forms of the Lopsided Lov\'{a}sz Local Lemma (LLLL) (an extension of the LLL to probability spaces in which the bad-events are ``negatively correlated'' in a certain technical sense; see \cite{erdos-spencer}). Some well-known applications of the LLLL which we treat here include random permutations and $k$-SAT. Our algorithms here are also much faster than \cite{hss}. Some applications of this technique are also given to partial Latin transversals, improving upon \cite{stein:latin}. 

\smallskip \noindent \textbf{(d) Entropy of the MT-distribution and combinatorial enumeration.} We show another concrete way in which the MT-distribution has significant randomness --  that its \emph{R\'{e}nyi entropy} \cite{chor-goldreich:weak-randomness} is relatively close to that of the initial product distribution. (The min-entropy is a special case of the R\'{e}nyi entropy and has become a central notion in randomness extractors and explicit constructions: see, e.g., \cite{cohen:two-source,cz,nisan-zuckerman:extractors,vadhan:pseudorandomness}.) For many applications of the LLL, such as $k$-SAT, non-repetitive coloring etc., this implies that the solution set has greater cardinality than was known before; perhaps more excitingly, it further builds on item (c) above to prove for the first time that MAX-SAT instances, as just one example, have several good solutions. 

\smallskip
To summarize, we consider some basic applications of the LLL, and develop (much) faster algorithms for these, some of which are the first-known polynomial-time- or Las-Vegas- algorithms. We also present improved/new algorithms and enumerative results in settings where we can allow a few bad events to happen. The impetus behind our work is further investigation of the MT-distribution
and some of its relatives. 

\subsection{Technical overview}
The original analysis of Moser \& Tardos gave sufficient conditions for their MT algorithm to terminate, yielding a configuration without bad-events. However, often one would like more information about such configurations, beyond the bare fact that they exist. As shown in \cite{hss}, one can define an \emph{MT-distribution}: the probability distribution induced on configurations that are output from the MT algorithm. The MT-distribution was used by \cite{hss} to show that in various MT applications, one can guarantee that the output of the MT algorithm has additional good properties.

Another useful application of this principle comes from \cite{harris-srin-focs}, which uses the MT disribution to find configurations (e.g. independent transversals) which have certain large-scale average properties as well. For example, one may define a weighting function on elements and find configurations with high overall weight, by examining the expected weight in the MT-distribution.

In this paper, we take the notion of the MT-distribution much further: not only can one analyze the probability distribution on the \emph{output} of the MT algorithm, but one can also analyze the distribution on its \emph{intermediate states}. These intermediate distributions share many properties with the original sampling distribution $\Omega$, which is just a product distribution. In particular, the key step of the MT algorithm --- the search for currently-true bad events --- is quite similar to a search problem over a random configuration. Random configurations are often easy to search: for example, while deciding $k$-colorability is NP-hard in general, a simple algorithm of \cite{decide-k-color} solves it for Erd\H{o}s-R\'{e}nyi random graphs in expected polynomial time.

The key step of the MT algorithm thus often boils down to finding a bad-event in a (nearly) random configuration. This can often be accomplished by \emph{branching algorithms}, in which one gradually builds up a putative true bad event by ``guessing'' successively more of its state. At every step, one can check whether the partial bad event is extendable to a full bad event, and abort the search if not. Using the randomness of the configuration, one can show that there is a good probability of aborting early.

\subsection{Outline}
In Section~\ref{mt-review-sec}, we review the analysis of the MT algorithm. We describe witness trees, a key proof-technique for showing the convergence of that algorithm, which also plays a key role in understand the MT distribution. We also introduce a new variant of the critical Witness Tree Lemma, which allows us to bound the probability of events in internal states of the MT algorithm.

Sections~\ref{a1fapp-sec:fast} describes our basic algorithms and data structures. Two applications are given, for Ramsey numbers and for hypergraph 2-coloring. They are good representatives of ``typical'' applications in combinatorics and algorithms, and they show how these techniques can lead to faster algorithms for many LLL applications, even those which already have polynomial-time algorithms.

Section \ref{a1fapp-sec:latin} analyzes a variant of the MT algorithm for random permutations, and shows that one can obtain the first sub-linear (square-root of input size) algorithms for Latin transversals, a problem of fundamental combinatorial interest.

Section~\ref{a1fapp-sec:exp-to-poly} addresses non-repetitive vertex coloring -- one of the few remaining cases where polynomial-time versions of the LLL were not known -- and develops such polynomial-time versions.

Section~\ref{a1fapp-sec:partial-avoid} addresses the problem of partially avoiding bad events, in cases where the LLL criterion is not satisfied. We tighten the bounds of \cite{hss}, giving a symmetric criterion in the case when $e p d = \alpha$, for $\alpha \in [1,e]$, as well as, for the first time, an asymmetric criterion. Furthermore, we give a faster parallel algorithm in this case; while applying the parallel MT algorithm directly, as in \cite{hss}, would give a running time of $O(\frac{\log^3 m}{(1 - \alpha)^2})$, we improve this to $O(\frac{\log^2 m}{1 - \alpha})$.

Section~\ref{a1sec:min} estimates the entropy of the MT-distribution, and shows that it is close to the original distribution. This automatically implies that there are many more solutions than known before for various problems such as $k$-SAT, non-repetitive coloring, and independent transversals -- and especially the maximum-satisfiability variants of these problems. 

\section{Witness trees and the MT-distribution}
\label{mt-review-sec}
The analysis of \cite{moser-tardos} is based on \emph{witness trees}, an analytical tool which provides the history of all variables that lead up to a resampling. These give an explanation or witness for each of the resamplings that occurs during the MT algorithm. As shown in \cite{hss}, these witness trees can also be used to give explanations for other types of events (not necessarily bad events). We will give a very brief overview of these results here; the reader should consult \cite{moser-tardos} and \cite{hss} for a much more in-depth explanation of these concepts.

Suppose we run the MT algorithm, and we resample the bad-events $B^1, \dots, B^T$ in order; the MT algorithm may or may not have terminated by this point. We may produce a \emph{witness tree} $\hat \tau^k$ for the $k^{\text{th}}$ resampling, as follows. We begin by placing a singleton root node labeled $B^k$. We then proceed backward for $t = k-1, k-2, \dots, 1$; for each bad-event $B^t$, we see if there are any nodes of $\hat \tau^k$ which are labeled by some $B' \sim B^t$. If there are not, then we do not modify $\hat \tau^k$. If there, we select one such node at greatest depth in $\hat \tau^k$, and attach to it a new leaf node labeled $B^t$. 

In this description, $\hat \tau^k$ is a random variable. One may also fix a specific labeled tree $\tau$, and examine if $\hat \tau^k = \tau$ for any value of $k$. If there is some value of $k$ for which $\hat \tau^k = \tau$, we say that $\tau$ \emph{appears}. To distinguish these related notions, we use the term ``tree-structure'' to refer to a particular labeled tree which could be produced as a value for the (random variable) $\hat \tau^t$.

The key lemma in \cite{moser-tardos}, which governs the behavior of the MT algorithm, is the Witness Tree Lemma:
\begin{definition}[Weight of a witness tree]
For any tree-structure $\tau$, whose nodes are labeled by events $B_1, \dots, B_s$, we define the \emph{weight} of $\tau$ by
$w(\tau) = \prod_{i=1}^s P_{\Omega}(B_i)$.
\end{definition}

\begin{lemma}[Witness Tree Lemma]
\label{wtlemma1}
For any tree-structure $\tau$, $P(\text{$\tau$ appears}) \leq w(\tau)$. 
\end{lemma}

One key result of \cite{moser-tardos} is the following:
\begin{proposition}[\cite{moser-tardos}]
\label{a1fapp-gen-tree-prop}
Let $B$ be any bad event. The total weight of all tree-structures rooted in $B$ is at most $\mu(B)$.
\end{proposition}

In \cite{hss}, Lemma~\ref{wtlemma1} and Proposition~\ref{a1fapp-gen-tree-prop} and were extended to arbitrary events. Given some event $E$ which occurs during the MT algorithm, one can build a ``witness tree'' for it. The tree has a root node, labeled by $E$; one constructs the remainder of the tree in the same manner as we have previously described, going backward in time and inserting nodes labeled by bad-events. These trees have a slightly different form to those analyzed by Moser \& Tardos; their root node is labeled by $E$, and all the other nodes are labeled by bad-events.

Given a tree-structure  $\tau$ rooted in $E$, we say that $\tau$ appears if $\hat \tau^k = \tau$, where $k$ is some time at which $E$ is true during the MT algorithm. The weight of such a tree, whose nodes are labeled by events $E_1, \dots, E_k$ (which are not all necessarily bad-events), is $\prod_{i=1}^k P_{\Omega}(E_i)$. The Witness Tree Lemma applies here as well:
\begin{proposition}[\cite{hss}]
Let $\tau$ be a tree-structure rooted in $E$. The probability that $\tau$ appears is at most $w(\tau)$.
\end{proposition}

In order to state the result of \cite{hss}, it will be convenient to have the following notation: for any event $E$, we define 
\begin{equation}
\label{a1feqn:theta}
\theta(E) = P_{\Omega}(E) \sum_{\substack{\mathcal I \subseteq N(E)\\\text{$\mathcal I$ independent}}} \prod_{B \in \mathcal I} \mu(B)
\end{equation}

Note that 
$$
\theta(E) \leq P_{\Omega}(E)  \prod_{B \sim E} (1 + \mu(B)) \leq P_{\Omega}(E) \exp( \sum_{B \sim E} \mu(B))
$$
for any event $E$, where $\exp(t)$ denotes $e^t$. Also, note that in the symmetric LLL setting, we have $\theta(E) \leq P_{\Omega}(E) \exp(e \cdot p \cdot |N(E)|)$. The asymmetric LLL criterion can be summarized compactly as $\mu(B) \geq \theta(B)$ for all $B$.

\begin{proposition}[\cite{hss}]
\label{hss-result}
Let $E$ be any event. The total weight of all tree-structures with a root node $E$, and the remaining nodes consisting of bad-events, is at most $\theta (E)$.  Hence, the probability that event $E$ occurs in the output of the MT-distribution is at most $\theta(E)$.\footnote{We note that in \cite{hss} a slightly weaker result was proved; this Proposition~\ref{hss-result} follows easily by combining Pegden's analysis \cite{pegden} and Bissacot et al.'s cluster-expansion criterion \cite{bissacot} with the ideas of \cite{hss}}
\end{proposition}

\subsection{A witness tree lemma for internal states}
We now introduce a key lemma which allows us to bound the probability of events occuring in internal states of the MT algorithm. One crucial feature of this lemma is that we can not only compute the probability that $E$ occurs, but we can count the number of times it occurs.
\begin{lemma}
\label{internal-lemma1}
Let $E$ be any event, and let $B \in \mathcal B$. Then
$$
\sum_{t=1}^T P( E(X^t) \wedge B^t = B) \leq \mu(B) \theta(E).
$$

(To clarify the notation, $E(X^t)$ means that event $E$ is true in the configuration $X^t$.)
\end{lemma}
\begin{proof}

For each time $t$ satisfying $E(X^t)$ and $B^t = B$, one may construct a type of witness tree which we denote $\hat \tau^t$. This is constructed in a similar manner to that of \cite{hss}. We place a node labeled by $E$ at the root and place a child node labeled by $B$ below it. (Note that we do not necessarily have $E \sim B$, and so the $B$ would not necessarily have been placed as a child of $E$ in the standard method for generating witness trees.) We then go backward in time through the execution log of the MT, placing any resampled bad events in the tree (as children of $E$ or $B$ or lower nodes). 

We refer to the set of possible witness trees that can be produced in this fashion as $E/B$-tree-structures.

We note that all the witness trees that are produced in this fashion are distinct; for, in the $k^{\text{th}}$ resampling of $B$, the witness tree $\hat \tau^t$ has $k$ nodes which have label $B$. This implies that
$$
\sum_{t=1}^T [E (X^t) \wedge B^t = B] \leq \sum_{\text{$E/B$-tree-structures $\tau$}} [ \text{$\tau$ appears} ]
$$
where $[E(X^t) \wedge B^t = B]$ is (here and throughout the paper) the Iverson notation, which is one if $E(X^t) \wedge B^t = B$ is true and zero otherwise.

Next, one may show that the witness tree lemma holds for $E/B$-tree-structures. Namely, for each fixed tree-structure $\tau$, we have $P(\text{$\tau$ appears}) \leq w(\tau)$. (The proof of this is nearly identical to Proposition~\ref{hss-result}.) Hence we have
$$
\sum_{t=1}^T P(E (X^t) \wedge B^t = B) \leq \sum_{\text{$E/B$-tree-structures $\tau$}} w(\tau)
$$

So let us consider the total weight of all such $E/B$-tree-structures. We define a mapping $f$ from pairs of tree-structures $\tau_1, \tau_2$ rooted in $E, B$ respectively to an $E/B$-tree $\tau = f(\tau_1, \tau_2)$. This mapping is defined by adding $\tau_2$ as a child of the root node of $\tau_1$. 

This mapping is surjective --- given an $E/B$-tree-structure $\tau$, which has a root node $E$ and a child node $v$ labeled $B$, let $\tau_2$ be the subtree rooted at $v$ and the let $\tau_1 = \tau - \tau_2$; then $f(\tau_1, \tau_2) = \tau$. Furthermore, this mapping has the property that $w( f(\tau_1, \tau_2) ) = w(\tau_1) w(\tau_2)$. Thus, we have that
\begin{align*}
\sum_{\text{$E/B$-tree-structures $\tau$}} w(\tau) &\leq \sum_{\substack{\text{tree-structures $\tau_1$}\\\text{rooted at $E$}}} \sum_{\substack{\text{tree-structures $\tau_2$}\\\text{rooted at $B$}}} w( f(\tau_1, \tau_2)) \\
&=\sum_{\substack{\text{tree-structures $\tau_1$}\\\text{rooted at $E$}}}   \sum_{\substack{\text{tree-structures $\tau_2$}\\\text{rooted at $B$}}}  w(\tau_1) w(\tau_2)
\end{align*}

By Proposition~\ref{hss-result}, we have $\sum_{\text{tree-structures $\tau$ rooted at $E$}} w(\tau) \leq \theta(E)$. By Proposition~\ref{a1fapp-gen-tree-prop}, we have $\sum_{\text{tree-structures $\tau$ rooted at $B$}} w(\tau) \leq \mu(B)$. Hence the total weight of all $E/B$-tree-structures is at most $\mu(B) \theta(E)$.

\end{proof}

\section{Fast search for bad events}
\label{a1fapp-sec:fast}
To implement the MT algorithm, we must search for any bad-events which are currently true (or certify there are none). The simplest way to do this would be to check the entire set $\mathcal B$ in each iteration. This will cost $\Omega(m)$ time per iteration (at least). If the bad-events are provided to us an arbitrary list, this is optimal. However, most applications of the LLL have more bad events than variables, and these bad events are much more structured.

Consider the very first iteration of the MT algorithm, searching for currently-true bad-events. In this case, the variables $X$ are distributed according to $\Omega$, a product distribution. For many problems, one can search random configuration faster (in expectation) than arbitrary configurations. Thus, one should be able to perform the first search step much faster than $\Omega(m)$ time. As the MT algorithm proceeds, the distribution becomes distorted. However, we prove that it does not stray too far from its original distribution. Thus, one can still hope to find bad-events significantly faster on these intermediate distributions than on arbitrary distributions.

For most applications of the MT algorithm, including all those in this paper, the remaining steps of the MT algorithm can be done relatively efficiently. For example, resampling each variable typically takes $O(1)$ time. As the work of resampling variables will always be negligible compared to finding true bad-events, we will ignore this cost throughout.

\subsection{Efficient search algorithms}
One main ingredient of our algorithms is a problem-specific search algorithm $S$ which given an assignment $X$ of the variables, determines all the bad-events currently true on $X$. This search procedure may be randomized, consuming a random source $R$ (which is independent of the random source used to drive the MT algorithm itself). We refer to this as $S(X,R)$.  

In many settings, finding a search algorithm which gives good worst-case bounds can be difficult or impossible. However, we will seek to parametrize the run-time of $S$ so that we can analyze its behavior on distributions drawn from the intermediate stages of MT. We thus define an \emph{event-decomposition} for $S$ to be a set of events $A_{i}$ (not necessarily bad events) and constant terms $c_{i}$, where $i$ ranges over the integers, with the property that 
\begin{equation}
\label{a1feqn:s}
\bE_R [ \text{Time($S(X, R)$)} ] \leq \sum_i c_{i} [A_{i} (X)].
\end{equation}
It is important to note in this definition that the expectation is taken only over the random source $R$ consumed by $S$, \emph{not} on the randomness of the MT process itself.

We can now measure the running time of MT as follows:
\begin{theorem}
\label{a1fapp-thm2}
Given an event-decomposition for $S$ as in (\ref{a1feqn:s}), define $T = \sum_i c_i \theta(A_i)$.
Then, $\bE[\text{run-time of MT}] \leq (1 + \sum_{B \in \mathcal B} \mu(B)) T$.
\end{theorem}
\begin{proof}
We sum over the times $t = 0, \dots, t-1$ so that
$$
\bE_R[ \text{time} ] \leq \sum_{t=0}^T \bE_R[ \text{Time($S(X^t, R)$)} ] \leq c_i \sum_{t=0}^T \sum_i P(A_i(X^t))
$$

We first consider time $t = 0$. The configuration $X^0$ has exactly the distribution $\Omega$, hence $P(A_i(X^0)) = P_{\Omega}(A_i(X)) \leq \theta(A_i)$.

Next, for each time $t = 1, \dots, T$ we have that
$$
\sum_{t=1}^T P(A_i(X^t)) =  \sum_{B \in \mathcal B}  \sum_{t = 1}^TP(A_i(X^t) \wedge B^t = B).
$$

By Lemma~\ref{internal-lemma1}, this is $\sum_{B \in \mathcal B} \mu(B) \theta(A_i)$. The result follows.
\end{proof}

\subsection{Example: Faster algorithms to construct Ramsey graphs}
\label{a1sec:ramsey}
A classical result in combinatorics is the lower bound on the diagonal Ramsey number $R(k,k) > \frac{\sqrt{2}}{e} k 2^{k/2}$ via the LLL \cite{alon-spencer}. This can be viewed also as an algorithmic challenge: given $k$, two-color the edges of the complete graph $K_n$ for $n =  \lceil \frac{\sqrt{2}}{e} k 2^{k/2} \rceil$, such that no $k$-clique has all $\binom{k}{2}$ edges of the same color.
\begin{proposition}[Follows straightforwardly from MT]
\label{a1ramsey-simple-prop}
For $n =  \lceil \frac{\sqrt{2}}{e} k 2^{k/2} \rceil$, there is an algorithm to construct a two-coloring of $K_n$ avoiding monochromatic $k$-cliques, in expected $2^{k^2/2 + o(k^2)}$ time.
\end{proposition}
\begin{proof}
For each $k$-clique, there is a bad-event that it is monochromatic; this has probability $p = 2^{1-\binom{k}{2}}$. There are $m = \binom{n}{k} \leq n^k / k!$ cliques, and so the expected number of resamplings is at most $m e p$. For each resampling, we check each $k$-clique, which takes $\binom{k}{2} m$ time. Thus, the total expected time in $O(e p \binom{k}{2} m^2) \leq 2^{k^2/2 + o(k^2)}$.
\end{proof}

Although there are exponentially many bad-events in this case, they have a combinatorial structure and it is not necessary to search each bad-event individually. Rather, we can use a type of branching algorithm to enumerate the cliques. This search algorithm was developed in \cite{harris2} in the context of a similar application of the LLL; however, in that case, it was only necessary to analyze the \emph{initial} configuration.

\begin{proposition}
\label{a1ramsey-search}
There is a deterministic search algorithm $S$ for monochromatic $k$-cliques with an event decomposition
$$
\text{Time}( S(X) ) = n^{O(1)} \negthickspace \negthickspace \sum_{\substack{\text{cliques $I$}\\ |I| \leq k}} [ \text{$I$ monochromatic on $X$} ]
$$
\end{proposition}
\begin{proof}
We recursively enumerate all $i$-cliques, for $i=2, \dots, k$. Initially, every edge is a monochromatic $2$-clique. Next, for each monochromatic $i$-clique $I$, we test all possible vertices $v$ and check if $I \cup \{v \}$ is also monochromatic. It takes $\binom{i}{2}$ time to check each $i$-clique, so the total time for this process (extending a given $i-1$ clique to $i$-cliques) is at most $O(n \binom{i}{2}) \leq n^{O(1)}$.
\end{proof}

\begin{proposition}
For $n =  \lceil \frac{\sqrt{2}}{e} k 2^{k/2} \rceil$, there is an algorithm to construct a two-coloring of $K_n$ avoiding monochromatic $k$-cliques, in expected $2^{k^2/8 + o(k^2)}$ time.
\end{proposition}
\begin{proof}
We apply Theorem~\ref{a1fapp-thm2} to the event-decomposition of Proposition~\ref{a1ramsey-search}. We have:
{\allowdisplaybreaks
\begin{align*}
T &= n^{O(1)} \smashoperator{\sum_{\substack{\text{cliques $I$}\\ |I| \leq k}}} \theta(  \text{$I$ monochromatic on $X$} ) \\
&\leq n^{O(1)} \sum_{i=2}^k \sum_{\text{$i$-cliques $I$}} 2^{1 - \binom{i}{2}} \exp(e p |N(I)| ) \\
&\leq n^{O(1)} \sum_{i=2}^k n^i 2^{- \binom{i}{2}} \exp(e p i^2 n^{k-2} / (k-2)! )  \\
&\leq 2^{k^2/8 + o(k^2)}
\end{align*}
}
Now, $\sum_{B} \mu(B) \leq m e p = 2^{O(k)}$. Hence by Proposition~\ref{a1fapp-thm2} the overall run-time of MT is $2^{k^2/8 + o(k^2)}$.
 \end{proof}

This is a polynomial improvement over Proposition~\ref{a1ramsey-simple-prop}, roughly reducing the time to the fourth root.

Many of our algorithms to search for bad-events have the same flavor as the search for Ramsey graphs: we want to find some structured bad-event, which involves many variables. Instead of seeking to enumerate over the entire set of variables at once, we build up the variables gradually. This leads to a type of branching process. At level $i$ of the process, we have ``guessed'' a set of $i$ variable indices; we then check whether it is possible that there is a bad-event involving them. If we can rule this out, we abort the branching process; otherwise we extend it by trying to add a new variable. We refer to each partial list of variables, which is putatively involved in a true bad-event, as a \emph{story}. For example, in the case of Ramsey graphs, a story is an $i$-clique for $i \leq k$.

\subsection{Depth-first-search Moser-Tardos}
\label{a1fapp-sec:dfs}
As we have seen, the main cost in the MT algorithm is to search for any bad-events which are currently true (or certify there are none). The simple way to do this, as we have discussed in Section~\ref{a1fapp-sec:fast}, is to check the entire set $\mathcal B$ in each iteration. This is rather wasteful; an optimization suggested by Joel Spencer, is to maintain a stack which records all the currently-true bad-events. At the very beginning of the MT algorithm, we scan the entire set $\mathcal B$ to find all the true bad-events. Whenever we resample a bad-event $B$, we only need to check its neighbors to determine whether they became true (and if so, we add them to the stack); we do not need to search the entire space. 

For example, in the symmetric LLL setting, we must expend $O(d)$ work after each each resampling (assuming that we have an adjacency list for the dependency graph and it requires unit time to check a bad-event). As the expected number of resamplings overall is $O(m/d)$, this gives a total expected running time $O(m)$. If the bad-events are simply provided to us as an arbitrary list, this is already optimal. 

We refer to this as a ``depth-first-search" MT. This can potentially improve the runtime of MT by up to a factor of $n$; because instead of needing to re-scan all the bad-events, we only need to scan those affected by the most-recently-resampled variables.

For applications with structured bad-events, we can speed up the depth-first search strategy by taking advantage of the random nature of the MT-distribution. We can hope to design a search algorithm which takes as input a configuration of variables, \emph{and a bad-event $B$}, and lists all of the bad events $B' \sim B$ which hold in it.

\smallskip \noindent \textbf{A key ingredient: data structure $D$.}
One main ingredient of our algorithms is a problem-specific data-structure $D$ which, given a bad event $B$ and a configuration $X$, can determine all the bad events $B' \sim B$ which may be caused to be true by resampling $X$. This data-structure also requires an initialization step, in which given a variable-assignment $X$ we find \emph{all} bad events currently true in it, as well as recording any other information about $X$ needed to use the data structure later. (Initialization is typically much cheaper and simpler than the updating step, and is only performed once, so we mostly ignore it in our analyses.)

In addition, we may want to use a randomized data-structure; we allow $D$ to uses a random bit-string $R$ (which is independent of the randomness used to drive the MT algorithm itself). This leads to the following formulation:
\begin{theorem}
\label{a1fapp-thm3}
Suppose that we are given an event-decomposition $\{c_{B,i}, A_{B,i} \mid B \in \mathcal B \}$ and a randomized data-structure $D$ which satisfies the following condition: 

Suppose that, given a bad-event $B$ and configuration $X$, the data-structure $D(B,X)$ finds all the bad-events which are true on $X$ and are dependent with $B$. Furthermore, for any fixed $B, X$ suppose we have
\begin{align*}
\bE_R \Bigl[ \text{Time}(D(B, X)) \Bigr] \leq \sum_i c_{B,i} [A_{B,i} (X)]
\end{align*}

For each event $B$, define $T_B = \sum_i c_{B,i} \theta(A_{B,i})$.

Then, the expected run-time of the MT algorithm, \emph{exclusive of time required for the initialization steps},  is at most $\sum_{B \in \mathcal B} \mu(B) T_B$.
\end{theorem}
\begin{proof}
We sum over time $t = 1, \dots T$:
\begin{align*} 
\bE \Bigl[ \sum_{t=1}^T \text{Time}(D(B^t, X^t)) \Bigr] &= \sum_{t=1}^T  \sum_i c_{B^t,i} P(A_{B^t,i} (X^t)) \\
&= \bE \Bigl[ \sum_{t=1}^T \sum_{B \in \mathcal B} \sum_i c_{B,i} [A_{B,i} (X^t) \wedge B^t = B] \Bigr] \\
&= \sum_{B \in \mathcal B}  \sum_i c_{B,i} \sum_{t=1}^T P(A_{B,i} (X^t) \wedge B^t = B) \\
&\leq \sum_{B \in \mathcal B} \mu(B) \sum_i c_{B,i} \theta(A_{B,i}) \qquad \text{by Lemma~\ref{internal-lemma1}} \\
&\leq \sum_{B \in \mathcal B} \mu(B) T_B
\end{align*}
\end{proof}

\subsection{Example: hypergraph two-coloring}
\label{a1fapp-sec:hyp2col}

We consider a more technically involved example. Suppose we are given a $k$-uniform hypergraph with $m$ hyper-edges, and we wish to find a two-coloring of the vertices so that no edge is monochromatic. For each edge $f$, let $N(f)$ denote the edges which intersect with $f$ (including $f$ itself). If $|N(f)| \leq L \leq 0.17 \sqrt{\frac{k}{\ln k}} 2^k$ for all edges $f$, then MT can be applied to the approach of \cite{radha} to find a good coloring. The analysis of \cite{radha} introduces a separate bad-event for each intersecting \emph{pair} of edges; thus, straightforward analysis would indicate a running time $m L \cdot \text{poly($k$)}$; potentially, a quadratic-time algorithm. (Another variant of that algorithm, given in \cite{cherkashin}, would lead to an analogous result.) We reduce this to $m \log^{O(1)} m$ time.

\textbf{Set-up for the LLL.} We begin by describing a version of the algorithm of \cite{radha} to find such a coloring via the LLL. First, each vertex chooses a color at random. Next, we choose a random ordering of the vertices (equivalently, each vertex independently chooses a random rank $\rho_v \in [0,1]$). For each vertex $v$ in this order, we look for any monochromatic edges of which $v$ is the lowest-ranking vertex. If we find any such edge, we flip the color of $v$. 

It is easy to implement this procedure in time $O(m)$, but the probability that it succeeds can be very low when $m \gg L$. We we will assume that $m \geq \Omega(\sqrt{\frac{k}{\log k}} 2^k)$; otherwise, as shown in \cite{radha}, then this algorithm produces a good coloring with probability $\Omega(1)$.

This procedure fails to produce a valid coloring only if the following occurs. There is some edge $f$, originally colored blue (w.l.o.g.), and vertex $v \in f$ is the lowest-ranking vertex of $f$. There is another edge $f'$, which intersects $f$ in exactly $v$, with the property that all other vertices in $f'$ are either red or have rank lower than $v$. In that case, it is possible that all the originally blue vertices in $f'$ are flipped, becoming red. This type of edge will remain monochromatic in the final coloring.

Each vertex has two variables associated with it: its (original) color and its rank $\rho_v$. We use the MT algorithm to select both values.

We will translate this into the LLL framework in a somewhat unusual way. We define a bad event 
$B^{\text{blue}}(f,f')$ to mean that the above event occurred \emph{and the minimum-ranking vertex in $f$ had rank $\leq R$}, where $R = \frac{\ln k}{2 k}$. We define a bad event $B^{\text{blue}}(f)$ to mean that edge $f$ was originally blue and \emph{all vertices in it had rank $> R$}. We similarly define $B^{\text{red}}(f)$ and $B^{\text{red}}(f, f')$. Note that the algorithm fails iff at least one of the four types of bad events occurs. The reason we are distinguishing the two cases of the minimum-ranking vertex in $f$, is that when this rank is large, then fixing $f$ will typically break many $f'$; so it is not beneficial to take a union-bound over all such $f'$.

We now use the asymmetric LLL. For an event $B(f)$, we assign $\mu(B(f)) = \sqrt{e} p_1$ and
for an event $B(f,f')$ we assign $\mu(B(f,f')) = e p_2$, where $p_1 = P_{\Omega}(B(f)), p_2 = P_{\Omega}(B(f,f'))$.

Let us first compute $p_1$. For an event $B^{\text{blue}} (f)$, it must occur that all the vertices in $f$ are blue and have rank $> R$; this occurs with probability $p_1 = 2^{-k} (1-R)^k$.

Next, let us compute $p_2$. Suppose $f, f'$ intersect in $v$. For an event $B^{\text{blue}} (f,f')$, it must occur that all vertices in $f$ are blue; this occurs with probability $2^{-k}$. All the vertices in $f$, other than $v$, must have rank exceeding that of $v$; this occurs with probability $(1-\rho_v)^{k-1}$. All the vertices in $f'$, other than $v$, must be either red or have rank less than $v$; this occurs with probability $(1/2 + 1/2 \rho_v)^{k-1}$. Hence, integrating over $\rho_v \in [0,R]$, we have
\begin{align*}
p_2 &\leq \int_{\rho_v = 0}^{R} d \rho_v ~2^{-k} (1-\rho_v)^{k-1} (1/2 + 1/2 \rho_v)^{k-1} \\
&= 2^{1-2 k} \int_{\rho_v = 0}^{R} d \rho_v~(1-\rho_v)^{k-1} (1 + \rho_v)^{k-1} \\
&\leq 2^{1-2 k} R
\end{align*}

Finally, we need to analyze the dependency. Consider an edge $f$; let us define
$$
t = \prod_{B} (1 + \mu(B))
$$
where $B$ ranges over all bad events touching $f$. One can verify there are at most $2 L$ events of type $B(f')$ (one for each color) and at most $4 L^2$ events of $B(f', f'')$ (either $f'$ or $f''$ could touch $f$, and there are two possible colors). Hence we have 
$$
t \leq (1 + \sqrt{e} p_1)^{2 L} (1 + e p_2 )^{4 L^2} \leq \exp(2 L \sqrt{e} p_1 + 4 L^2 e p_2)
$$

The LLL criterion is now 
$$
p_1 \sqrt{e} \geq p_1 t \qquad p_2 e \geq p_2 t^2
$$
which can be seen to be satisfied for $L \leq 0.17 \sqrt{\frac{k}{\ln k}} 2^k$ and $k$ sufficiently large. In this case also we have $t \leq O(1)$.

\textbf{A data-structure to find bad-events.} Now that we have formulated this problem for the LLL, we come to the core algorithmic challenge: finding bad-events efficiently. For this, we will need a data-structure $D$ to track the following information: for each vertex $v$, we use a doubly-linked list to enumerate all \emph{monochromatic} edges which contain $v$.

For any edge $f$ and vertex-coloring $X$, we let $A(X,f)$ be the event that $f$ is monochromatic on $X$.
\begin{proposition}
\label{data-structure-event}
The data-structure $D$ allows us to find bad-events with an event-decomposition
$$
D(B,X) \leq k^{O(1)} \sum_{f \sim B} \Bigl( \sum_{g \in N(f)} 1 + \sum_{g' \in N(g)} \bigl( [A(g,X)] + [A(g',X)] \bigr) \Bigr)
$$
\end{proposition}
\begin{proof}
To simplify the notation, we write $f \sim B$ if $f$ is involved in $B$; that is, if $B$ is of the form $B(f)$ or $B(f,f')$.

First, we consider the cost to update the list of monochromatic edges. If an edge $f$ was originally monochromatic and is resampled, we delete it from the $k$ corresponding vertex-lists; that takes time $O(k)$. If an edge $f$  becomes monochromatic, we add it to the $k$ corresponding lists, again in time $O(k)$. The only edges which can change their status are those intersecting $B$, and so this is at most $\sum_{f \sim B} k$.

Next, we show how to find the bad events caused by resampling some edge $f$. To find an event of type $B(g)$ affected by $f$, we simply loop over all the monochromatic edges $g$ intersecting $f$, and check if they also satisfy the property that $\rho(w) \geq R$ for all $w \in g$; this takes time $\sum_{g \in N(f)} k^{O(1)}$.

Next, we search for events $B(g, g')$ in the configuration $X$, where $g \in N(f)$: we begin by looping over all edges $g \in N(f)$. If $g$ is monochromatic on $X$, we loop over all $g' \in N(g)$ and check whether $B(g,g')$ is true on $X$. The total work for this is
$$
k^{O(1)} \Bigl( \negthickspace \negthickspace  \sum_{g \in N(f)} \negthickspace \negthickspace  1 + [A(g,X)] |N(g)| \Bigr)
$$

Finally, consider how to find an event $B(g, g')$, where now $g' \in N(f)$. We begin by looping over $g' \in N(f)$; for each such edge $g'$, we want to find any edges $g$ where $B(g,g')$ is true. Let $G(g')$ denote the edges $g \in N(g')$ which are monochromatic on $X$.  We make the critical observation we can use our data-structure to enumerate, for each $v \in g'$, all the monochromatic edges including $v$, and so each $g \in G(g')$ is listed at most $k$ times. Thus, the total work to enumerate $G(g')$ is at most $k |G(g')|$; this is potentially much smaller than $N(g')$. Hence, the work for this step is
$$
k^{O(1)} \Bigl(  \sum_{g' \in N(f)} \negthickspace  1 + |G(g')| \Bigr)
$$

Putting all these terms together, we have that the total work expended searching for bad-events caused by resampling $f$ is at most
\begin{align*}
\text{Time} &\leq  k^{O(1)} \Bigl( \sum_{g \in N(f)} 1 + |G(g)| + [A(g,X)] |N(g)|  \Bigr) \\
&= k^{O(1)} \Bigl( \sum_{g \in N(f)} 1 + \sum_{g' \in N(g)} \bigl( [A(g,X)] + [A(g',X)] \bigr) \Bigr)
\end{align*}

Summing over all $f \sim B$, we have that
$$
D(B, X) \leq k^{O(1)} \sum_{f \sim B} \Bigl( \sum_{g \in N(f)} 1 + \sum_{g' \in N(g)} \bigl( [A(g,X)] + [A(g',X)] \bigr) \Bigr)
$$
\end{proof}

\begin{proposition}
The expected total time for the MT algorithm to find a coloring is at most $m k^{O(1)}$.
\end{proposition}
\begin{proof}
We apply Theorem~\ref{a1fapp-thm3} to the event-decomposition of Proposition~\ref{data-structure-event}.  For any bad-event $B(f)$, we have
\begin{align*}
T_{B(f)} \leq  k^{O(1)} \Bigl( L + \sum_{g \in N(f), g' \in N(g)} \theta(A(g)) + \theta(A(g')) \Bigr)
\end{align*}

For any edge $g$, we have $P_{\Omega}(A(g)) = 2^{-k}$ and so $\theta(A(g)) \leq P_{\Omega}(A(g)) \times t \leq O(2^{-k})$. Thus, we have that
\begin{align*}
T_{B(f)} &\leq  k^{O(1)} \Bigl( L + \sum_{g \in N(f), g' \in N(g)} O(2^{-k}) + O(2^{-k}) \Bigr) \leq  k^{O(1)} ( L + 2^{-k} L^2 ) \leq L k^{O(1)}
\end{align*}

Hence, the total expected work for this bad event $B(f)$, over the entire execution of MT, is at most $\mu(B(f)) T_B \leq p_1 \sqrt{e} L k^{O(1)} \leq k^{O(1)}$; summing over all edges $f$ gives a total time of $m k^{O(1)}$.

A similar argument applies to estimate $T_{B(f,f')} \leq m k^{O(1)}$ and to bound the time required to initialize the data structure. Recalling that $k = \log^{O(1)} m$, this proves the theorem.
\end{proof}

\section{Latin transversals}
\label{a1fapp-sec:latin}
Suppose we are given an $n \times n$ matrix $A$, in which each cell is assigned a color. Suppose that each color appears at most $\Delta \leq (27/256) n$ times in the matrix. We wish to select a permutation $\pi \in S_n$ with the property that no color appears twice, that is, there are no distinct $x, x'$ with the property that $A(x, \pi(x)) = A(x', \pi(x'))$. Such a permutation is referred to as a \emph{Latin transversal}; see \cite{bissacot,erdos-spencer,lllperm} for some of the long history behind this and related notions. 

One can apply the Lopsided LLL to the probability space defined by a random permutation. In this context, a bad-event is that we have $\pi(x) = y \wedge \pi(x') = y'$ where $A(x,y) = A(x',y')$. In \cite{erdos-spencer}, it is shown that two events are dependent for this probability space (in the sense of the lopsided LLL) iff they overlap in a row or column of the matrix.

In \cite{lllperm}, a variant of the MT algorithm was presented for finding such permutations in polynomial time. The algorithm is somewhat complicated to describe, but the basic idea of this algorithm is that one can resample bad-events by performing \emph{random swaps} of the relevant permutation entries. These random swaps play the same role as a resampling in the usual MT algorithm.

Although this algorithm and its analysis are much more complicated than the standard MT algorithm, one can still develop witness trees and show that  Witness Tree Lemma holds. This implies that all the results about the MT-distribution do as well. This is one of the key advantages of the proof-technique developed in \cite{lllperm}; later works, such as \cite{achlioptas} and \cite{harvey}, have developed substantially simpler and more general proofs of the convergence of the swapping MT algorithm, but these approaches do not extend to the MT-distribution results.

\begin{theorem}
\label{a1fapp-thm-latin}
Suppose each color appears at most $\Delta \leq (27/256) n $ times in the matrix $A$. Then there is an algorithm to find a Latin transversal in expected time $O(n)$ assuming that we have fast read access to the matrix, namely:
\begin{enumerate}
\item[(A1)] The entries of $A$ allow random-access reads.
\item[(A2)] The colors of $A$ can be represented as bit-strings of length $O(\log n)$.
\item[(A3)] Our algorithm can perform elementary arithmetic operations on words of size $O(\log n)$ in time $O(1)$. 
\end{enumerate}

Note that the input size to the problem is $\Theta(n^2)$. 
\end{theorem}
\begin{proof}
Each bad event $B$ has probability $p = \frac{1}{n(n-1)}$. It is shown in \cite{lllperm} that the asymmetric LLL criterion holds with these parameters and that $\mu(B) = O(p)$ for any bad-event $B$. For any $x, y \in [n]$ and any bad-event $B$, we say that $B$ \emph{involves} $x$ or $y$ if $B$ contains a bad-event containing $\pi(x) = y'$ or containing $\pi(x') = y$. We define $w(x,y) = \prod_{\text{$B$ involves $x$ or $y$}} (1 + \mu(B))$.

We can enumerate such events as follows: there are $2 n-1$ choices for the first cell involving column $x$ or row $y$, and $\Delta \leq O(n)$ choices for the other cell with the same color. So there are $O(n^2)$ such bad events, and for each such bad event $B$ we have $\mu(B) = O(n^{-2})$, so in total $w(x,y) = O(1)$.

Now consider the following data-structure $D$. We first choose some pairwise-independent hash function $H$, uniformly mapping the labels of colors to the set $[n]$ \cite{carter}. We will maintain a list, for each $t \in [n]$, of all pairs $(x,y)$ with $\pi(x) = y$ and $H(A(x,y)) = t$. These can be maintained with a doubly-linked list for each element $t \in [n]$ in the range of $H$. We will update this structure during the execution of the Swapping Algorithm; for example, if $\pi(x) = y$ and we resample to a new permutation $\pi'$ with $\pi'(x) = y'$, we would remove the pair $(x,y)$ from the list corresponding to $H(A(x,y))$ and add the pair $(x, y')$ to the list corresponding to $H(A(x,y'))$. It is not hard to see how to add and remove pairs from their appropriate list in constant time. 

Now consider the work required in a single step of $D(B, X)$. The operation of adding and removing pairs from their corresponding linked-lists takes $O(1)$ time. The costly operation is that, for each affected position $x$ in the permutation, we must loop over all pairs $x, x'$ with $H(A(x, \pi(x))) = H(A(x', \pi(x')))$ and test whether $A(x, \pi(x)) = A(x', \pi(x'))$. If the latter holds, then we have detected a new bad event. 

Thus, suppose we resample $B = (\pi(x_1) = y_1) \wedge (\pi(x_2) = y_2)$, obtaining the new permutation $\pi'$. There are four positions in the permutation $\pi'$ that differ from $\pi$, and we must test each of these to see if there are new bad events. Thus, the time to update $D$ is given by
\begin{align*}
\sum_{y_1' \in [n]} \sum_{\substack{x_3 \neq x_1\\ y_3 \neq y_1'}} \Bigl[ \pi'(x_1) = y'_1 \wedge \pi'(x_3) = y_3 \wedge H(A(x_1, y_1')) = H(A(x_3, y_3)) \Bigr] + \cdots
\end{align*}
(Here, we have only written one of the four summands, corresponding to new bad events involving $\pi(x_1) = y_1'$. The other three summands are analogous, and will have the same cost.)

By $2$-independence of $H$, we have that the \emph{expected time} to update $D$ from a bad-event $B$ is 
\begin{align*}
\sum_{y_1' \in [n]} \sum_{\substack{x_3 \neq x_1\\ y_3 \neq y_1'}} \Bigl[ \pi'(x_1) = y'_1 \wedge \pi'(x_3) = y_3 \Bigr] \times \Bigl( 1/n + \bigl[ A(x_1, y'_1) = A(x_3, y_3) \bigr] \Bigr) + \cdots
\end{align*}

This expectation is taken over the hash function $H$, \emph{not} on any of the random choices during the MT algorithm. Thus, the permutations $\pi, \pi'$, should be viewed as fixed values and not random variables.

We can now apply Theorem~\ref{a1fapp-thm3} to calculate:
\begin{align*}
T_B &= \sum_{\substack{y_1', x_3 \neq x_1, y_3 \neq y_1'}}  \theta (\pi'(x_1) = y_1' \wedge \pi'(x_3) = y_3)  \Bigl(1/n + [A(x_1,y'_1) = A(x_3, y_3)] \Bigr) \\
&\leq \sum_{\substack{y_1', x_3 \neq x_1, y_3 \neq y_1'}}  P_{\Omega} (\pi'(x_1) = y_1' \wedge \pi'(x_3) = y_3) w(x_1,y_1') w(x_3, y_3)  \bigl(1/n + [A(x_1,y'_1) = A(x_3, y_3)] \bigr)
\end{align*}

Using the fact that there are at most $\Delta n = O(n^2)$ values of $y'_1, x_3, y_3$ with $A(x_1,y'_1) = A(x_3, y_3)$, and our bounds $w(x,y) \leq O(1)$, we calulate that this $T_B \leq O(1)$.

Thus, the expected running time of MT is 
\begin{align*}
\sum_B \mu(B) T_B  &\leq O(1) \negthickspace \negthickspace \negthickspace \negthickspace \sum_{\substack{x,y,x', y' \\ A(x,y) = A(x', y')}} \negthickspace \negthickspace \negthickspace \negthickspace  \mu(x,y,x',y') = O(n).
\end{align*}

A similar calculation shows an $O(n)$ time to initialize $D$.
\end{proof}

\section{Non-repetitive vertex coloring: from exponential to polynomial}
\label{a1fapp-sec:exp-to-poly}
So far, we have examined problems in which good data structures can lead to polynomial improvements in the MT runtime. However, Theorems~\ref{a1fapp-thm2},  \ref{a1fapp-thm3} are much more powerful, and can indeed transform exponential-time algorithms to polynomial-time ones. We will consider a series of related problems based on \emph{non-repetitive vertex coloring} of graphs. These represent some of the few remaining cases in which the LLL provides a proof of existence, but for which we do not know corresponding polynomial-time algorithm. 

Given a graph $G$, we seek to color its vertices so that no color sequence appears repeated in any vertex-simple path; i.e., there is no simple path colored $x x$, where $x$ can denote any nonempty sequence of colors. How many colors are needed in order to ensure such a coloring exists? This is known as the \emph{Thue number} $\pi(G)$ of $G$, motivated by Thue's classical result that $\pi$ is at most $3$ for paths of any length \cite{thue:1906}.\footnote{There are a few variants on this definition such as whether the edges or vertices are colored, and whether each has its own palette of colors or whether there is a common palette. For concreteness, we color vertices from a common palette; all of our bounds would apply to the other scenarios as well. We assume that the graph $G$ is simple with $2 \leq \Delta \leq n - 1$. }

The problems of finding non-repetitive colorings and Thue numbers have been studied extensively in a variety of contexts.
In \cite{alon2003}, it was shown via the LLL that for any graph $G$ with maximum degree $\Delta$, $\pi(G) = O(\Delta^2)$. The original constant term in that paper was not tight; a variety of further papers such as \cite{gryt1,gryt2,haranta} have brought it down further. The best currently-known bound is that $\pi(G) \leq (1 + o(1)) \Delta^2$ \cite{dujmovic}. 
The analysis of \cite{dujmovic} does not use the LLL; it uses a \emph{non-constructive} Kolmogorov-complexity argument which is somewhat complicated and specialized to the graph-coloring problem. 

While the MT resampling framework applies to this problem, the key bottleneck is to either find a bad event (a path with repeated colors), or to certify that none such exists. In this case, the number of bad events is exponentially large; more seriously, it is NP-hard to even detect whether a given coloring has a repeated color sequence \cite{marx:discrete09}. So, in this situation it is \emph{intractable} to find a data-structure for finding bad-events with good \emph{worst-case} run-time bounds.

In \cite{hss}, a constructive algorithm was introduced using $C= \Delta^{2+\epsilon}$ colors (i.e., if a slack $\Delta^{\epsilon}$ is allowed). The basic idea of \cite{hss} is to apply the MT algorithm, but to ignore the long paths. This algorithm succeeds in finding a good coloring with high probability, \footnote{We say an event occurs with high probability (abbreviated whp) if it occurs with probability $1 - n^{-\Omega(1)}$.} and the running time is $n^{O(1/\epsilon)}$ -- polynomial time for fixed $\epsilon$. This cannot be amplified to succeed with probability 1, as it is not clear how to test whether the output of the algorithm is a good coloring. Thus, it is a Monte Carlo, but not a Las Vegas, algorithm. 

\subsection{New results}
We present the first polynomial-time coloring that shows $\pi(G) \leq (1 + o(1)) \Delta^2$; furthermore, our algorithm is Las Vegas. 
Until this work, no Las Vegas algorithms were known for this problem where the number of colors $C$ is \emph{any} function of $\Delta$, and no Monte Carlo algorithms were known where $C = \phi \Delta^2$ for $\phi$ any fixed constant. We also develop the first-known $ZNC$ (parallel Las Vegas) versions of such results. 

As another application, Section~\ref{a1fapp-sec:higher-order-Thue} considers a generalization of non-repetitive colorings, introduced in \cite{alon-grytczuk}, to avoid \emph{$k$-repetitions}. That is, given an integer parameter $k \geq 2$, we aim to color the vertices to avoid the event that a sequence of colors  $x x \dots x$ appears on a vertex-simple path, with the string $x$ occurring $k$ times.  (Standard non-repetitive coloring corresponds to $k=2$.) The best type of result achievable in polynomial time using \cite{hss} is a coloring using $O(\Delta^{2 + \epsilon})$ colors, for any desired \emph{constant} $\epsilon > 0$. 
Theorem~\ref{a1fthm:higher-order-thue} gives a Monte Carlo algorithm to find a coloring using 
$C = \Delta^{1 + \frac{1+\epsilon}{k-1}} + O(\Delta^{2/3 + \frac{1+\epsilon}{k-1}})$ colors and which avoids any $k$-repetitions, running in $n^{O(1/\epsilon)}$ (i.e., polynomial) time.

A second type of generalization of non-repetitive colorings comes from work of \cite{krieger}, which considered when it is possible to avoid nearly-repeated color sequences; that is, a sequence of colors $x y$ where the Hamming distance of $x$ and $y$ is small. The work of \cite{krieger} considered the problem for coloring paths. In Section~\ref{a1fapp-sec:col-seq} while we extend this to general graphs. This presents new algorithmic challenges as well.

\subsection{Non-repetitive vertex coloring}
\label{a1fapp-sec:thue-2}

\begin{proposition}
\label{non-rep-prop1}
There is some constant $\phi > 0$, such that for any graph $G$ of maximum degree $\Delta$, there is a non-repetitive vertex coloring with $C = \Delta^2 + \phi \Delta^{5/3}$ colors.
\end{proposition}
\begin{proof}
We show this via the LLL. A bad-event in this context is some vertex-simple path with a repeated color sequence, of length $2 l$. We define $\mu(B) = \alpha^{2 l}$ for all such events, where $\alpha$ is a parameter to be determined. Our convention is that each color sequence gives rise to a distinct bad-event; thus, all bad-events are \emph{atomic} and have probability $C^{-2 l}$.

Now consider a fixed vertex $v$, and let us consider the sum $\mu(v)$ over all bad-events $B$ which involve vertex $v$. Such bad-events have the following form: There is a path of length $2 l$, of which $v$ is the $t^{\rm th}$ vertex for some $t = 0, \dots, l -1$ (by reversing the path, one can assume without loss of generality $v$ comes in the initial half); the first $l$ vertices have some pattern of colors, and the final $l$ vertices have also this pattern. 

Summing over all possible values of $t, l$, all $\Delta^{2 l - 1}$ paths, and all possible $C^l$ color patterns, we have 
\begin{align*}
\mu(v) &\leq \sum_{l = 1}^{\infty} \sum_{t = 1}^{l} C^{l} \Delta^{2 l - 1} \alpha^{2 l} \\
&\leq \frac{\alpha^2 C \Delta}{(1 - \alpha^2 C \Delta^2)^2} \qquad \text{for $\alpha^2 C \Delta^2 < 1$}
\end{align*}

To show that the asymmetric LLL criterion holds, consider some bad-event $B$ defined by a path $v_0, \dots, v_{2 l - 1}$.  Its probability is $C^{-2 l}$. Its independent sets of neighbors can be determined by, for each $i = 0, \dots, 2 l - 1$,  selecting zero or one bad-events involving $v_i$. Thus, we have that
$$
\sum_{\substack{I \subseteq N(B) \\ \text{$I$ independent}}} \prod_{B' \in I} \mu(B') \leq \prod_{i=0}^{2 l - 1} (1 + \mu(v_i)) \leq (1 + \frac{\alpha^2 C \Delta}{(1 - \alpha^2 C \Delta^2)^2})^{2 l}.
$$ 

Thus, the LLL criterion becomes
$$
\alpha^{2 l} \geq C^{-2 l}(1 + \frac{\alpha^2 C \Delta}{(1 - \alpha^2 C \Delta^2)^2})^{2 l}
$$
which is satisfied for all $l \geq 1$ iff
\begin{equation}
\label{a1fnon-rep-eqn}
\alpha C \geq 1 + \frac{\alpha^2 C \Delta}{(1 - \alpha^2 C \Delta^2)^2}
\end{equation}

Set $\alpha = (\sqrt{C} (\Delta + \Delta^{2/3}))^{-1}$; routine algebra shows that (\ref{a1fnon-rep-eqn}) holds for $\phi$ sufficiently large.
\end{proof}

The challenge is to turn this exisential proof into an efficient algorithm. The key bottleneck is to search for some true bad event; we will do so via Theorem~\ref{a1fapp-thm3}. The following intermediate result will be useful. (Recall the definition of $\theta$ from (\ref{a1feqn:theta})

\begin{proposition}
\label{theta-e-prop}
Suppose we have any event $E$ of the form $\chi(v_1) = c_1 \wedge \chi(v_2) = c_2 \wedge \dots \wedge \chi(v_k) = c_k$, where $v_1, \dots, v_k$ are distinct vertices and $c_1, \dots, c_k$ are color labels. Then we have that
$$
\theta(E) \leq \alpha^k
$$
where $\alpha =  (\sqrt{C} (\Delta + \Delta^{2/3}))^{-1}$.

Suppose we have any event $E'$ of the form $\chi(v_1) = \chi(u_1) \wedge \dots \wedge \chi(v_k) = \chi(u_k)$, where $v_1, \dots, v_k, u_1, \dots, u_k$ are distinct vertices. Then we have
$$
\theta(E') \leq \beta^k
$$
where $\beta = (\Delta + \Delta^{2/3})^{-2}$.
\end{proposition}
\begin{proof}
The event $E$ has probability $P_{\Omega}(E) = C^{-k}$. To form an independent set of neighbors of $E$, one may select, for each $i = 1, \dots, k$, one or zero path including $v_i$. We have already computed this sum in Proposition~\ref{non-rep-prop1}, and so we have that the sum over all such independent sets is at most $\prod_{i=1}^k (1 + \mu(v_i)) \leq (1 + \frac{\alpha^2 C \Delta}{(1 - \alpha^2 C \Delta^2)^2})^k$.

Because the LLL criterion is satisfied, we have that this is at most $(\alpha C)^k$. Thus, overall we have
$$
\theta(E) \leq C^{-k} \times (\alpha C)^k = \alpha^k
$$

The bound on $E'$ follows by taking a union bound over all possible colors $c_1, \dots, c_k$ and computing the probability that $\chi(v_1) = c_1 = \chi(u_1) \wedge \dots \wedge \chi(v_k) = c_k = \chi(u_k)$.
\end{proof}

In Theorem~\ref{a1fnon-rep-thm}, we will show via Theorem~\ref{a1fapp-thm3} that the coloring can be found in $O(n^2)$ time using the DFS MT algorithm. As a warm-up exercise, we begin with a slightly weaker result; we use Theorem~\ref{a1fapp-thm2} to produce the coloring in $\text{poly}(n)$ time.  

\begin{theorem}
\label{a1fnon-rep-thm0}
The coloring of Proposition~\ref{non-rep-prop1} can be found in expected time $O(n^3 \Delta^{4/3})$.
\end{theorem}
\begin{proof}
We construct a search algorithm to find bad-events which are currently true. We suppose that $C \leq n$, as otherwise this is trivial (assign each vertex a distinct color)

To begin, we sort all the neighborhoods of every vertex by color. As the number of colors is $O(n)$, then this step can be implemented in $O(n^2)$ time.

Now, suppose we want to find a vertex sequence $v_0, \dots, v_{2 l -1}$ of length $2 l$, where $l$ is fixed. We construct a branching process for $i = 0, \dots, l-1$, wherein in stage $i$ we enumerate over possible values for $v_i, v_{i+l}$. In order for these correspond to a bad-event, it must be that $\chi(v_i) = \chi(v_{i+l})$. Furthermore, $v_i, v_{i+l}$ must be neighbors of $v_{i-1}, v_{i+l-1}$ respectively (unless $i = 0$). Finally, all the vertices $v_0, \dots, v_{2 l - 1}$ must be distinct.

Because we have sorted the adjacency lists of all the vertices by color, then for $i > 0$ and a fixed sequence $v_0, \dots, v_{i-1}, v_l, \dots, v_{i+l-1}$ one can enumerate over $v_i, v_{i+l}$ in time
$$
\sum_{v_i \in N(v_{i-1})} (1 + \sum_{v_{i+l} \in N(v_{i+l-1})} [ \chi(v_{i+l}) = \chi(v_i)] )
$$

(In this sum, and all the sums we encounter, we enforce the requirement that the vertices are distinct; we do not write this explicitly in simplify the notation.)

Summing over all possible choices for $v_0, \dots, v_{i-1}, v_l, \dots, v_{i+l-1}$, the overall time is given by
$$
\sum_{v_0, \dots, v_{i-1}, v_l, \dots, v_{i+l-1}} [\chi(v_0) = \chi(v_l) \wedge \dots \chi(v_{i-1}) = \chi(v_{i+l-1})] \sum_{v_i \in N(v_{i-1})} (1 + \sum_{v_{i+l} \in N(v_{i+l-1})} [ \chi(v_{i+l}) = \chi(v_i)] )
$$

Similarly, for $i = 0$, we can do this in time
$$
\sum_{v_0} (1 + \sum_{v_{}} [ \chi(v_{l}) = \chi(v_0)] )
$$

Thus, summing over $i = 0, \dots, l-1$ and $l = 0, \dots, n$, we have an event decomposition of the form
\begin{align*}
\text{Time} &\leq n^2 + \sum_{l=0}^n \sum_{v_0} (1 + \sum_{v_{l}} [ \chi(v_{l}) = \chi(v_0)] ) \\
&\qquad + \sum_{i=1}^{l-1} \sum_{v_0, \dots, v_{i-1}, v_l, \dots, v_{i+l-1}} [\chi(v_0) = \chi(v_l) \wedge \dots \chi(v_{i-1}) = \chi(v_{i+l-1})] \\
& \qquad \qquad \Bigl(  \sum_{v_i \in N(v_{i-1})} (1 + \sum_{v_{i+l} \in N(v_{i+l-1})} [ \chi(v_{i+l}) = \chi(v_i)] ) \Bigr)
\end{align*}

We evaluate $T$ as in Theorem~\ref{a1fapp-thm2}. For each value of $l$, the term $\sum_{v_0} (1 + \sum_{v_{l}} [ \chi(v_{l}) = \chi(v_0)] )$ contributes $n + \sum_{v_0, v_l} \theta( \chi(v_0) = \chi(v_l) )$; by Proposition~\ref{theta-e-prop}, the latter has value at most $n^2 \beta$.

Similarly, each of the terms 
$$
\sum_{v_0, \dots, v_{i-1}, v_l, \dots, v_{i+l-1}} [\chi(v_0) = \chi(v_l) \wedge \dots \chi(v_{i-1}) = \chi(v_{i+l-1})] \sum_{v_i \in N(v_{i-1})} (1 + \sum_{v_{i+l} \in N(v_{i+l-1})} [ \chi(v_{i+l}) = \chi(v_i)] )
$$
contributes $n^2 \Delta^{2 i-1} \beta^i + n^2 \Delta^{2 i} \beta^{i+1}$.

Summing over $l, i$, we have
\begin{align*}
T &\leq n^2 + \sum_{l=0}^n ( n^2 \beta + \sum_{i=0}^{l-1} n^2 \Delta^{2 i - 1} \beta^i + n^2 \Delta^{2 i} \beta^{i+1})  \\
&\leq O(n^2) (1 + \sum_{l=0}^n  \sum_{i=0}^{l-1} \beta^{i} \Delta^{2 i +1} ) \\
&\leq O(n^2) (1 + n \sum_{i=0}^{\infty} \beta^{i} \Delta^{2 i +1} ) \\
&= O(n^2 \Delta^{4/3})
\end{align*}

Next, observe that the total sum of $\mu(B)$ over all $B \in \mathcal B$ is at most $\sum_v \sum_{\text{$B$ involves $v$}} \mu(B) \leq n \alpha C \leq O(n)$. Thus, the overall time is at most $(1 + \sum_B \mu(B)) T \leq O(n) \times O(n^2 \Delta^{4/3})$.
\end{proof}

We want to emphasize the intuition here, which is that searching for a repetitive coloring in the intermediate configurations of the MT algorithm is very similar for searching for a repetitive coloring in a completely random configuration. One could compute the expected running time of this branching algorithm on such a random coloring. This would give identical formulas, with the only difference being that all instances of $\alpha$ in the above proof would be replaced by the slightly smaller value $C^{-1}$, the probability that a given vertex has a given color.

We next improve on this by using depth-first search for MT, as well as being slightly more careful in our search algorithm.
\begin{theorem}
\label{a1fnon-rep-thm}
The coloring of Proposition~\ref{non-rep-prop1} can be found in expected time $O(n^2)$.
\end{theorem}
\begin{proof}
We assume throughout that $\Delta \leq \sqrt{n}$, as otherwise this is trivial (simply assign each vertex a unique color).

We will maintain a data structure $D$ in which we maintain the adjacency list of each vertex sorted by color. This costs $O(n^2)$ to initialize.

Suppose we are given a bad-event $B$, which is a path of vertice $w_0, \dots, w_{2 k-1}$ which is repetitively colored. In order to apply the depth-first-search MT algorithm, we must update $D$ identify any bad-events involving any vertices $w_0, \dots, w_{2k-1}$. We shall first show how, given a single vertex $v$, one can update $D$ identify any bad-events events involving $v$. We shall construct an event-decomposition such that
$$
\text{Time for vertex $v$} \leq \sum_{\text{events $E$}} c_{v,E} [E(\chi)] 
$$
where $\chi$ is the coloring after resampling $B$ and $c_{v,E}$ are non-negative constants. 

For each such vertex $v$, let us define
\begin{equation}
\label{tv-eqn}
T_v = \sum_{\text{events $E$}} c_{v,E} \theta(E)
\end{equation}

Then by Theorem~\ref{a1fapp-thm3}, we have
$$
T_B \leq T_{w_1} + T_{w_2} + \dots + T_{w_{2 k-1}}
$$

So, in order to bound $T_B$, it suffices to show an upper bound on $T_v$, for a given vertex $v$.

Thus, suppose we are given a configuration and a fixed vertex $v$, and we wish to update $D$ and determine if $v$ participates in any paths with repeated colors.  We begin by updating the sorted adjacency lists for each neighbor of $v$; this takes time $O(\Delta^2)$. 

Next, say that $v$ participates in a repeated path $v_0, \dots, v_{2 l -1}$ of length $2 l$, and occurs in position $t < l$. For the moment, let us suppose that $t=0$ and $l$ is fixed. To emphasize the position of $v$ in the list, we write $v_t = v = v_0$.

We will use a branching process similar to Theorem~\ref{a1fnon-rep-thm0}, in which a story corresponds to a list of distinct vertices $v_0, v_1, \dots, v_i, v_t, v_{l+1}, \dots, v_{l+i}$ for some $i = 0, \dots, l$. 

We begin by looping over the vertex in position $l$, restricting the search to vertices $v_l$ which has the same color as $v_0$. We also loop over all neighbors $v_1, v_{l+1}$ of $v_0, v_l$ respectively. Again, if they have the same color (and also $v_1 \neq v_{l+1}$), then we continue the search otherwise we abort. 
We continue this process, looping over pairs of vertices $v_2, \dots, v_{l-1}, v_{l+2}, \dots, v_{2 l -1}$. At each stage of this branching process, we insist that the colors in the path are repeated up to that point, and all vertices are distinct. At the end, we examine if the resulting path corresponds to a bad event. We can do a similar procedure if $t \neq 0$; we begin by guessing vertices $v_{t+1}, \dots, v_{l-1}, v_{t+l}, \dots, v_{2 l-1}$ and then branch backward on $v_{t-1}, \dots, v_0, v_{l+t-1}, \dots, v_l$.

As in Theorem~\ref{a1fnon-rep-thm0}, we can perform this enumeration in overall time
\begin{equation}
\label{tv-eqn2}
\Delta \times \Bigl( \sum_{v' \neq v} [\chi(v') = \chi(v)]  +  \smashoperator{\sum_{\substack{v' \neq v\\w \in N(v), w' \in N'(v)}}} [\chi(v') = \chi(v) \wedge \chi(w) = \chi(w')] \Bigr)
\end{equation}
where here the terms $w, w'$ indicate potential candidates for $v_1, v_{l+1}$ and $v'$ is a potential candidate for $v_l$.

By Proposition~\ref{theta-e-prop}, the overall contribution of this expression is at most to (\ref{tv-eqn}) is at most
$$
\Delta \times \Bigl( \sum_{v' \neq v} (\Delta + \Delta^{2/3})^{-2}   + \smashoperator{\sum_{\substack{v' \neq v\\w \in N(v), w' \in N'(v)}}} (\Delta + \Delta^{2/3})^{-4} \Bigr)
$$
which is $O(n \Delta^{-1})$.

Continuing in this way, we see that the $r^{\text{th}}$ level of this branching process has overall contribution to (\ref{tv-eqn}) of $O(n \Delta^{2 r +1} \beta^{r+1})$.

With a little thought, one can see that it is not necessary to specify a fixed value of $l, t$ for this branching. Once one specifies the initial vertex $v_t$ (without necessarily knowing $t$) and the corresponding vertex $v_{t+l}$ (again, without necessarily knowing $l$), one merely has to decide how many steps to branch forward/backward from these two vertices. If at some point during this branching process one detects a repeated color sequence, one can then infer the corresponding $t, l$. 

If one branches $r_1$ forward steps and $r_2$ backward steps, then the contribution of the resulting work factor to $T_v$ is similarly 
$$
O(n \Delta^{2(r_1 + r_2) + 1} \times \beta^{r_1 + r_2 + 1})
$$
Summing over $r_1, r_2$, one has the total work for $v$ is at most 
$$
T_v \leq \Delta + O( \sum_{r_1 = 0}^{\infty} \sum_{r_2 = 0}^{\infty} n \Delta^{2(r_1 + r_2) + 1} \times \beta^{r_1 + r_2 + 1})
$$ a simple calculations shows this is at most $O(\Delta + n \Delta^{-1/3}) \leq O(n)$.

This bound on $T_v$ yields a bound on $T_B$ for any bad-event $B$ which is a path of length $2 l$:
$$
T_B \leq 2 l \times O(n)
$$

Summing over all such bad events, we have
\begin{align*}
&\sum_{B} \mu(B) T_B  \leq \sum_{l=1}^{\infty} n \Delta^{2 l - 1} C^l \alpha^{2 l} \times 2 l \times O(n) \leq O(n^2)
\end{align*}

\end{proof}

\subsection{Parallel algorithm for the Thue number}
Moser \& Tardos introduced in \cite{moser-tardos} a generic parallel form of their resampling algorithm. This algorithm can be summarized as follows:
\begin{enumerate}
\item[1.] Draw $X_1, \dots, X_n$ from $\Omega$.
\item[2.]  Repeat while there is some true bad event:
\begin{enumerate}
\item[3.] Choose (arbitrarily) maximal independent set $I$ of currently-true bad events $B \in \mathcal B$.
\item[4.] Resample all the bad-events $B \in I$ in parallel.
\end{enumerate}
\end{enumerate}

As shown in \cite{harris-haeupler}, this algorithm will terminate with high probability after $O( \frac{\log n}{\epsilon} )$ rounds, as long as we satisfy a slightly stronger form of the LLL criterion, namely we satisfy it with $\epsilon$-slack. That is, for each bad-event $B$ we require
$$
\mu(B) \geq (1+\epsilon) \theta(B)
$$
for some $\epsilon > 0$. Furthermore if we can detect the currently-true bad-events in time $O(\log^2 n)$, then the overall running time is $O(\frac{\log^3 n}{\epsilon})$. 

In order to turn this into an efficient randomized algorithm, it suffices to enumerate at each stage all currently-true bad-events, using polylogarithmic time and polynomial space. (This automatically implies that there are a polynomial number of true bad-events, and so a maximal independent set of them can be found efficiently via Luby's algorithm.)\footnote{Alternatively, \cite{moser-tardos} shows that the parallel algorithm terminates after $O(\frac{\log \sum_{B \in \mathcal B} \mu(B)}{\epsilon})$ iterations, and one may show directly in this case that this is $O(\frac{\log n}{\epsilon})$. The analysis of \cite{harris-haeupler} shows this directly without needing to compute $\sum_B \mu(B)$.}

\begin{proposition}
\label{a1fnon-rep-prop2}
There is a constant $\phi > 0$ such that any graph $G$ of maximum degree $\Delta$ can be $C$-colored to avoid repetitive vertex-colorings as long as
$C \geq \Delta^2 + \phi \Delta^{2}/\log \Delta$. 
Furthermore, such a coloring can be found in ZNC (Las Vegas NC): the algorithm terminates successfully with probability 1 after expected time $O(\log^4 n)$ using $\text{poly}(n)$ processors.
\end{proposition}
\begin{proof}
Along the same lines as Theorem~\ref{a1fnon-rep-thm}, a sufficient condition for the parallel MT algorithm with $\epsilon$ slack is
\begin{equation}
\label{a1fnon-rep-eqn2}
C \alpha - \frac{\alpha^2 C \Delta}{(1 - \alpha^2 C \Delta^2)^2} - n T C^T \Delta^{2 T} \alpha^{2 T} \geq 1 + \epsilon
\end{equation}
and this is satisfied for $\alpha = (\Delta^2 + \frac{\phi \Delta^2}{2 \log \Delta})^{-1}$. 

For $\phi, x$ sufficiently large, the LHS of (\ref{a1fnon-rep-eqn2}) is a decreasing function of $\Delta$, hence reaches its minimum value at $\Delta = n$. At this point, one can observe that (\ref{a1fnon-rep-eqn2}) is satisfied for $\epsilon = \Omega(1/\log n)$. Thus MT terminates after $O(\log^2 n)$ iterations whp.

Our task becomes to develop a branching process for finding currently-true bad-events, whose expected number of active stories is bounded by a polynomial and whose running time is polylogarithmic.

We will use a branching which proceeds through $l = 1, 2, \dots, \log_2 n$ rounds. At each round $l$, we enumerate all sets of vertices $v_0, \dots, v_{k-1}, w_0, \dots, w_{k-1}$ which satisfy the following conditions:
\begin{enumerate}
\item[(B1)] $k \leq 2^l$
\item[(B2)] $\chi(v_0) = \chi(w_0), \dots, \chi(v_{k-1}) = \chi(w_{k-1})$
\item[(B3)] $v_0, \dots, v_{k-1}, w_0, \dots, w_{k-1}$ are distinct.
\item[(B4)] $v_0, \dots, v_{k-1}$ and $w_0, \dots, w_{k-1}$ are paths.
\end{enumerate}
To extend the set of stories from stage $l$ to stage $l+1$, we use the following observation: if $v_0, \dots, v_{k-1}, w_0, \dots, w_{k-1}$ satisfy these conditions at stage $l+1$, then $v_0, \dots, v_{k/2-1}, w_0, \dots, w_{k/2-1}$ and $v_{k/2}, \dots, v_{k-1}, w_{k/2}, \dots, w_{k-1}$ both satisfy these conditions (separately) for stage $l$. Thus, we may build the set of all stories satisfying these conditions recursively by pairing stories at stage $l$ and checking if they survive to stage $l+1$.

Furthermore, we see that if there are $V_l^t$ stories satisfying these conditions at each time $t$ and stage $l$, then for each $l$ this pairing requires time $V_l^2 \text{poly}(n)$ and time $O(\log n)$. Thus, if we show that $V_l \leq \text{poly}(n)$ for each $l = 0, \dots, \log_2 n$ then this shows that this process can be implemented using $O(\log^2 n)$ time and $\text{poly}(n)$ processors.

Next, we claim that it suffices to show that $\bE[V^t_l] \leq \text{poly}(n)$. For, suppose that $\bE[V^t_l] \leq n^r$. Then by Markov's inequality we have that whp $V^t_l \leq n^r \times T \times \log_2 n \times n^{100}$. Furthermore, one may easily detect if $V_l$ exceeds this bound; if so, we abort the algorithm and start from scratch.

Finally, we turn to estimating $\bE[V^t_l]$. Given any fixed sequence $v_0, \dots, v_{k-1}, w_0, \dots, w_{k-1}$ satisfying (B1), (B3), (B4), we may slightly modify the proof of Proposition~\ref{theta-e-prop} to see that the probability that it satisfies (B2) as well is at most $\beta^k$ for
$$
\beta = C \alpha^2
$$

Now, in a manner similar to Theorem~\ref{a1fnon-rep-thm0}, we may take a union bound over all $k = 1, \dots, 2^l$ and all vertices $v_0, \dots, v_{k-1}, w_0, \dots, w_{k-1}$ satisfying (B1), (B3), (B4) to see that $\bE[V^t_l] \leq \text{poly}(n)$.

Thus, the overall expected running time is $O(\frac{\log^3 n}{\epsilon}) = O(\log^4 n)$ using a polynomial number of processors.
\end{proof}

\subsection{Higher-order Thue numbers}
\label{a1fapp-sec:higher-order-Thue}
Recall the notion of \emph{$k$-repetitions} introduced in \cite{alon-grytczuk}. That is, given a parameter $k$, we want to avoid the event that a sequence of colors  $x x \dots x$ appears on a vertex-simple path, with the string $x$ occurring $k$ times. 

It is not hard to extend the analysis of Theorem~\ref{a1fnon-rep-thm} to obtain an algorithm for $k$-Thue number as follows:
\begin{theorem}
For some constant $\phi > 0$, there is a Las-Vegas algorithm which takes as input a graph $G$ and parameter $k$, and produces a vertex coloring with $C =  \Delta^{1 + \frac{1}{k-1}} + \phi \Delta^{2/3 + \frac{1}{k-1}}$ colors which avoids $k$-repetitions. This algorithm runs in expected time $n^{k+O(1)}$.
\end{theorem}

For any fixed value of $k$, this is a polynomial-time algorithm. But developing an algorithm whose running time scales with $k$, presents new algorithmic challenges. Note that the approach of \cite{hss}, which is based on finding a ``core'' set of bad events which can be checked quickly, will not work here --- for, the work required to check even the color sequences of length $1$ (the simplest class of bad event), is already $n \Delta^k$, which can be super-polynomial time.

Our main result here is: 
\begin{theorem}
\label{a1fthm:higher-order-thue}
For some constant $\phi > 0$, there is an algorithm with the following properties. It takes as input a graph $G$, a parameter $k$, and a parameter $\epsilon$. It runs in expected time $n^{O(1/\epsilon)}$, and produces a vertex coloring with $C = \Delta^{1 + \frac{1+\epsilon}{k-1}} + \phi \Delta^{2/3 + \frac{1+\epsilon}{k-1}}$ colors, which avoids any $k$-repetitions whp. That is, there is no vertex-simple path in which a color sequence is repeated $k$ times. Note that this is \emph{not} a Las-Vegas algorithm.
\end{theorem}
\begin{proof}
Suppose we are given a fixed $\epsilon > 0$. As in Theorem~\ref{a1fnon-rep-thm}, for any bad-event $B$ of length $k l$, we set $\mu(B) = \alpha^{k l}$, where $
\alpha = \bigl( \Delta^{1 + \frac{1+\epsilon}{k-1}} + \frac{\phi}{2} \Delta^{2/3 + \frac{1+\epsilon}{k-1}} \bigr)^{-1}
$
Now observe that for $\phi > 0$, we have $\alpha^k C \Delta^k < 1$, so the LLL criterion reduces to
\begin{equation}
\label{a1flll-k-thue-eqn}
C \alpha \geq 1 + \frac{k \alpha^k C \Delta^{k-1}}{(1 - \alpha^k C \Delta^k)^2}
\end{equation}

The LHS of (\ref{a1flll-k-thue-eqn}) can be written as a function of $\Delta, k, \phi$, and a parameter $v = \Delta^{\epsilon/(k-1)}$.  By routine calculus, we see that this is indeed satisfied, for all $k, \Delta$, for $\phi$ sufficiently large. (The worse case comes when $k$ is small, $v = 1$, and $\Delta \rightarrow \infty$). Routine calculations show that this satisfies the LLL criterion for $\phi$ sufficiently large.

The remaining task is to find any bad events which are true in a current configuration. To begin, we will simply ignore any color-sequences whose length is greater than some threshold $L = x (\frac{\log n}{\epsilon \log \Delta})$ for some sufficiently large constant $x$. We claim that, even though we do not check these events explicitly, the probability that any such bad event ever becomes true, is negligible. For the probability that there is such a long path is at most $\sum_{\text{$B$ has length $l \geq L$}} \theta(B) \leq \sum_{l = L}^{\infty} n C^l \Delta^{k l} \alpha^{k l}$; routine analysis shows that this is $n^{-\Omega(1)}$. So we only need to check the shorter sequences.

Now, suppose we wish to check for a $k$-repetition involving a color sequence of length $l$. As we are not attempting to determine exactly the exponent of $n$, we will simplify our task by using Theorem~\ref{a1fapp-thm2}, searching the entire graph for repeated color sequences. We will also simply enumerate over the exact value of the length $l$ of the path, rather than attempting to handle all values of $l$ simultaneously. These simplifications are both wasting work but only by a factor of $n^{O(1)}$. 

We begin by guessing the full $l$-long color sequence. Once this color sequence $c_0, \dots, c_{l-1}$ is fixed, we use a branching process; a story at stage $i$ consists of the vertices $v_0, \dots, v_i$ in order, which agree with the color sequence (that is, $v_i$ has color $c_{i \text{mod l}}$). 

Let us consider the overall cost of this branching process. At the $i^{\text{th}}$ level of this process, we must enumerate over colors sequences $c_1, \dots, c_l$ and possibilities for the vertices $v_0, \dots, v_i$. Thus, we may write the cost as
$$
\text{Cost of $i^{\text{th}}$ level} \leq \sum_{c_0, \dots, c_{l-1}} \sum_{\substack{v_0, v_1 \in N(v_0), v_2 \in N(v_1), \dots \\ \text{$v_0, \dots, v_i$ distinct}}} [\chi(v_0) = c_0 \wedge \chi(v_1) = c_1 \wedge \dots ]
$$

This event-decomposition is in the appropriate form to apply Theorem~\ref{a1fapp-thm2}. By Proposition~\ref{theta-e-prop} (using a different definition of $\alpha$), we have $\theta( \chi(v_0) = c_0 \wedge \chi(v_1) = c_1 \wedge \dots \wedge \chi(v_i) = c_i) \leq \alpha^{i+1}$. As there are $C^l$ choices for the colors $c_0, \dots, c_{l-1}$ and $n \Delta^i$ choices for the vertices $v_0, \dots, v_i$, the total contribution of this expression is at most $n \Delta^i \alpha^{i+1}$. Thus, summing from $i = 0, \dots, k l$, we see that overall we have that the overall cost to find bad-events of length $l$ is at most $C^l \sum_{i=0}^{k l} n \Delta^i \alpha^{i+1} \leq n^{O(1)} C^l$.

As we are only examining color sequences of length at most $L$, the expected work overall is at most $T \leq n^{O(1)}  C^L \leq n^{O(1/\epsilon)}$.

It is notable in this proof that we need to combine the method of \cite{hss}, which is based on identifying a core subset of bad events, with the fast-search method of Theorem~\ref{a1fapp-thm2}. In this application, the large bad events cannot be searched efficiently; searching the small ``easy'' bad events efficiently takes exponential time in general but is polynomial time on the random configurations presented during the MT algorithm.
\end{proof}

\subsection{Approximately-repeated color sequences}
\label{a1fapp-sec:col-seq}
In \cite{krieger}, the idea of non-repeated color sequences was generalized to avoiding $\rho$-similar color sequences, for some parameter $0 < \rho \leq 1$. If $x, y$ are two color-sequences of length $l$, we say that $x, y$ are $\rho$-similar if $x, y$ agree in at least $\lceil \rho l \rceil$ positions. 
When $\rho = 1$, of course, this simply means that $x = y$. Hence the problem of coloring the graph to avoid $\rho$-similar color sequences generalizes the problem of non-repetitive coloring.  Although the work of \cite{krieger} considered the problem for color sequences alone, this generalization has not been studied in the context of graph coloring. It presents new algorithmic challenges as well. We present the following result:
\begin{theorem}
\label{a1fthm:rho-similar}
There is some constant $\phi > 0$ with the following property. For all $\rho \in (0,1]$ and any graph $G$ with maximum degree $\Delta$, there is a coloring that avoids $\rho$-similar sequences, with
$$C = \rho^{-1} (1-\rho)^{1 - 1/\rho}  (\Delta^2 + \phi \Delta^{11/6})^{1/\rho}$$
colors. Furthermore, such a coloring can be found in expected time $n^{O(1)}$.
\end{theorem}
\begin{proof}
Define the usual entropy function $h = h(\rho) = -(1 - \rho) \ln(1 - \rho) - \rho \ln \rho$. 

We can enumerate the bad events as follows. If we have a sequence $s$ of $2 l$ vertices, and a $l$-dimensional binary vector $w$ which has Hamming weight $H(w) = \lceil \rho l \rceil$, we define the bad event $B_{w,s}$ which is that vertices $s_i, s_{i+l}$ have the same color for all indices $i$ which $w_i = 1$. It is not hard to see that there is an $\rho$-similar vertex sequence iff there is some $w, s$ where the bad event $B_{w,s}$ occurs. (We can further insist that the vector $w$ has $w_1 = 1$; this gives slightly better bounds but does not change the asymptotics).

Set $\mu(B) = \alpha^{2 l}$ for a bad-event of length $2 l$, where $\alpha = e^{-h/\rho} (\Delta^2 + \frac{\phi}{2} \Delta^{11/6})^{-1/\rho}$

Let us count the bad events involving a vertex $v$. We enumerate this as follows. There are $(2 l) \Delta^{2 l - 1}$ paths involving vertex $v$. We must check a vector $w \in \{0, 1 \}^l$ which has a $1$ in the position corresponding to vertex $v$; this gives us $\binom{l-1}{\lceil \rho l \rceil - 1}$ further choices. Then there are $C^{\lceil \rho l \rceil}$ choices for the color sequence shared by $x, y$. Any such event has probability $\alpha^{2 \lceil \rho l \rceil}$. Summing over all $l$ gives us a total contribution of 
{\allowdisplaybreaks
\begin{align*}
\smashoperator[r]{\sum_{\text{$B$ involves $v$}}} \mu(B) &\leq \sum_{l = 1}^{\infty} (2 l) \Delta^{2 l -1} \binom{l-1}{\lceil \rho l \rceil - 1} \alpha^{2 \lceil \rho l \rceil} C^{\lceil \rho l \rceil}\\
&=\sum_{k = 1}^{\infty} (\alpha^2 C)^k \sum_{l=\lceil k/\rho \rceil}^{\lceil (k+1)/\rho \rceil - 1} (2 l) \Delta^{2 l -1} \binom{l-1}{k - 1} \\
& \leq \frac{2 \alpha^{2 \rho} \Delta e^h}{ (1 - \alpha^{2 \rho} C^{\rho} \Delta^2 e^h)^2 }
\end{align*}
}
Hence the asymmetric LLL criterion for avoiding such $\rho$-similar edge colors reduces to
$$
C \alpha -  \frac{2 \alpha^{2 \rho} C^{\rho} \Delta e^h}{ (1 - \alpha^{2 \rho} C^{\rho} \Delta^2 e^h)^2 } \geq 1
$$

Routine calculus shows that the LHS is decreasing in $\rho$. So the worst case is when $\rho = 1$; then simple calculus shows that this is satisfied for $\phi$ sufficiently large.

We now come to the main algorithmic challenge: finding a bad event (if any are currently true). One might naively expect to apply the branching process of Theorem~\ref{a1fnon-rep-thm}: first choose the first and middle vertex in the path. Then branch on the vertices, aborting the search early if the color sequence so far has too many disagreements. To see why this naive branching process does not give a polynomial-time algorithm, observe that we will not be able to remove \emph{any} stories in the early stages of the branching, because we might have a color sequence $x y$ in which the agreeing positions all come at the \emph{end}. Thus, the collection of stories will increase exponentially before collapsing exponentially. Although the number of final stories is relatively small, the intermediate story counts can become large. We want the agreeing positions to come fast enough to keep the number of stories small throughout.

We will branch on the color sequence starting not from the vertices at positions $0, l$ (the first and middle vertex in the path), but rather starting at positions $i, l+i$ for some well-chosen $i = 0, \dots, l-1$. At the $t^{\rm{th}}$ stage of the branching process, we will branch on the vertices at positions $i+t, l+i+t \text{ modulo $2 l$}$. Here, $t = 0$ corresponds to the initial choice of vertices, and $t = 1$ corresponds to choosing the first edge emanating from them.  At stage $t$ of the branching, we insist that the number of agreeing positions seen so far, is at least $\lceil t \rho \rceil$; otherwise we remove that possibility from the branching process.

To summarize, we use the following algorithm to find bad color sequences of length $2 l$:
\begin{enumerate}
\item[1.] For $a = 0, \dots, l-1$ repeat the following:
\begin{enumerate}
\item[2.] Initialize with a single, null story.
\item[3.] For $t = 0, \dots, l-1$ do the following:
\begin{enumerate}
\item[4.] For each story in the stack, count the number of positions at which the color sequences agree so far. If this number is smaller than $\lceil \rho t \rceil$, remove the story from the stack.
\item[5.] For each story remaining in the stack, choose the vertex at positions $(a+t) \text{ modulo $2 l$}$ and $(l + a+t) \text{ modulo $2 l$}$. 
Extend each story in all valid ways.
\end{enumerate}
\end{enumerate}
\end{enumerate}

We will first show that the running time for this algorithm is polynomially bounded. Let us fix some value of $a, t$, and consider the expected number of surviving stories. These must correspond to vertex paths of length $t$ whose color sequences agree on at least $\lceil \rho t \rceil$ positions. There are $\Delta^{2 t} n^{O(1)}$ choices for the vertices. For a fixed path, we can bound the probability that they agree on $\lceil \rho t \rceil$ positions as at most 
\begin{align*}
&\binom{t}{\lceil \rho t \rceil} c^{\lceil \rho t \rceil} \alpha^{2 \lceil \rho t \rceil} \leq n^{O(1)} e^{h t} c^{\rho t} \alpha^{2 \rho t} \leq n^{O(1)} \Bigl( \frac{\Delta^2 + \phi \Delta^{11/6}}{(\Delta^2 + (\phi/2) \Delta^{11/6})^2} \Bigr)^t
\end{align*}

Hence, the total expected number of stories for given $a, t$ is at most
$$
n^{O(1)} \Delta^{2 t-1} \Bigl( \frac{\Delta^2 + \phi \Delta^{11/6}}{(\Delta^2 + (\phi/2) \Delta^{11/6})^2} \Bigr)^t \leq n^{O(1)}
$$

Next, we must show that any bad event will indeed be discovered by this branching process. For, suppose $x, y$ are color sequence of length $l$ which agree on $\rho' l \geq \lceil \rho l \rceil$ positions. For $i = 1, \dots, l$ define $s_i$ to be the total number of agreements in positions $1, .., i$; for $i$ outside this range, define $s_i := s_{i \text{ mod } l}$. We also define the parameter $r_i = s_i - \rho' i$. Because $x, y$ agree on exactly $\rho' l$ positions, the sequence $r$ is periodic with period $l$.

We claim that for the value of $a$ in the range $1, \dots, l$ which minimizes $r_a$, then the color sequence $x y$ will survive the corresponding branching process. For, suppose at stage $t$, we lose $x y$. This implies that the total number of agreements between stages $a, a + t$ is strictly less than $\lceil \rho t \rceil \leq \rho' t$.  This implies that $s_{t+a} < s_{a} + \rho' t$ and hence $r_{t+a} < r_t$, contradicting minimality of $a$.
\end{proof}

\section{Partially avoiding bad events}
\label{a1fapp-sec:partial-avoid}
When the LLL condition is satisfied, then it is possible to select the variables so that \emph{no} bad events occur.  Alternatively, if one simply selects the underlying variables from $\Omega$ directly, then each bad event $B$ occurs with probability $P_{\Omega}(B)$. However, there can be a middle ground. As described in \cite{hss} even when the LLL condition is violated, one can use the MT-distribution to select the variables so that many fewer bad events occur than one would expect from $\Omega$. For example, if in the symmetric LLL setting we have $e p d = \alpha$, for $\alpha \in [1, e]$, then one can show that it is possible to cause at most $(1 + o(1)) m p e \ln(\alpha)/\alpha$ events to occur; here $o(1)$ is parameter which decreases with the dependency $d$ \cite{hss}. 

The result of \cite{hss} is based on the following idea: select each event to be a ``core event'' independently with probability $q$. These core events will not be allowed to occur; the non-core events are ignored. Each core event has on average $d q$ core neighbors. For $d$ sufficiently large, one can apply Chernoff bounds and the MT algorithm to ensure that the number of core neighbors is close to $d q$. Now, apply the MT algorithm a second time to avoid the core events, and show that in the MT distribution the non-core events have a high probability of being avoided.

While the method of \cite{hss} is intriguing, it suffers from a few shortcomings. First, the result is asymptotic; there is a second-order term, which is difficult to compute explicitly, and only goes away as $d \rightarrow \infty$. Second, this algorithm may be computationally expensive; the first application of the LLL, in particular, may dominate the second, ``real'' application, and may even be exponential time. Third, one obtains only gross bounds on the total number of true bad events; one cannot easily get more detailed information on the average behavior of a particular bad event.

In this section, we give new bounds and algorithms for partially avoiding bad events, which avoid these problems. In many cases, these algorithms are faster than the Moser-Tardos algorithm itself. The basic idea parallels \cite{hss}, in that we mark each bad event $B$ as core with probability $q(B)$. However, instead of using two separate LLL phases, we combine them into a single one. 

Recall the definition of $\theta(\cdot)$ from (\ref{a1feqn:theta}). 

\begin{theorem}
\label{a1falpha-resamp-thm1}
Suppose we are given a mapping $\mu: \mathcal B \rightarrow [0, \infty)$. Then there is an algorithm, which we refer to as the \emph{Truncated Moser-Tardos Algorithm}, whose output distribution $\Omega'$ on the underlying variables $X_1, \dots, X_n$ has the property
\begin{equation}
\label{a1fy1}
\forall B \in \mathcal B, P_{\Omega'} (B) \leq \max(0, \theta(B)-\mu(B))
\end{equation}
This algorithm has the same running-time behavior as other Moser-Tardos applications. In particular, the expected number of resamplings of a bad event is $\mu(B)$. 
(Note that the LLL criterion is simply that the RHS of (\ref{a1fy1}) is equal to zero.)
\end{theorem}

\begin{proof}
Given our original set of bad events $\mathcal B$, we define a new binary variable $Y(B)$ for each bad event, which is Bernoulli-$q(B)$ and which represents that $B$ is ``core". We introduce a new set of bad events $\mathcal B'$, defined as follows: for each bad event $B \in \mathcal B$, we define $B' \in \mathcal B'$ to be the event that $B$ is true and $Y(B) = 1$, where we define $q(B) = \min(1, \frac{\mu(B)}{\theta(B)})$. The truncated MT algorithm for $\mathcal B$ is then defined by running the MT algorithm for $\mathcal B'$.

It is not hard to see that the set of bad events $\mathcal B'$ satisfies the asymmetric LLL criterion with the weighting function $\mu$. 

Now, consider a bad event $B$. In order for $B$ to occur in the output, it must be the case that $Y(B) = 0$. Thus, we have that $P_{\Omega'}(B) = P_{\Omega'}(B \wedge (Y(B) = 0))$. We now apply Proposition~\ref{hss-result} so that $P_{\Omega'}(B \wedge (Y(B) = 0)) \leq \theta(B \wedge (Y(B) = 0)) = \theta(B) P_{\Omega}(Y(B) = 0) = (1 - q(B)) \theta(B)$.  By our choice of $q(B)$, this is $\max(0, \theta(B)-\mu(B))$.
\end{proof}

This specializes easily to the symmetric setting by setting $\mu(B) = (e/\alpha)^{1/d} - 1$ for all $B$:
\begin{corollary}
\label{a1fsym-cor-most}
Suppose each bad event $B$ has $P_{\Omega}(B) \leq p, |N(B)| \leq d$; and suppose that $epd \leq \alpha$ for $\alpha \in [1,e]$.
Then one can efficiently construct from a probability space $\Omega'$ in which each bad event $B$ occurs with probability at most $\frac{\ln \alpha}{d}$. The expected number of total resamplings is $O(m/d)$ to draw from $\Omega'$.
\end{corollary}

\subsection{Applications} 
As an example of the asymmetric form of Theorem~\ref{a1falpha-resamp-thm1}, consider $k$-SAT instances where each variable may appear in up to $L$ clauses in total (positively or negatively). Applying the Lopsided LLL, it is shown in \cite{gst} that $L \leq \frac{2^{k+1}}{e (k+1)}$ implies that the instance is satisfiable. We prove that this can be relaxed so that the instance is partially satisfiable.\footnote{One may verify that Theorem~\ref{a1falpha-resamp-thm1} holds for the variable-assignment LLLL, in which bad-events are dependent iff they disagree on a variable.}
\begin{theorem}
\label{a1fk-sat-most}
Suppose we have a $k$-SAT instance with $m$ clauses, in which each variable appears in up to $L \leq \frac{\alpha 2^{k+1}}{e k} - 2/k$
clauses (in total, either positively or negatively), for $\alpha \in [1,e]$. Then we can construct in expected time $m \log^{O(1)} m$ a truth assignment whose expected number of satisfied clauses is at least $m (1 - 2^{-k} e \ln(\alpha)/\alpha)$.
\end{theorem}
\begin{proof}
We assume that $m \geq 2^{k-1}$ as otherwise a randomly chosen solution will satisfy all the clauses with probability $1/2$, and the result follows trivially.

Suppose a variable $x_i$ appears in $l_i$ clauses; of these occurrences, it appears $\delta_i l_i$ positively and $(1-\delta_i) l_i$ negatively. Then, following the counter-intuitive choice described in \cite{gst}, we set variable $i$ to be T with probability $1/2 - x (\delta_i - 1/2)$, where $x \in [0,1]$ is a well-chosen parameter.

We set $\mu(B) = z$ for all bad events $B$, where $z$ is a parameter to be chosen. In this case, it suffices to show that
\begin{equation}
\label{ss1}
\forall B \in \mathcal B,  -z + P_{\Omega}(B) \exp( \smashoperator{\sum_{B' \sim B}} z ) \leq 2^{-k} e \ln \alpha/\alpha
\end{equation}

It is not hard to show, following \cite{gst}, that for $x = L z/2$ the LHS here is maximized when variables corresponding to the bad event $B$ each occur in exactly $L/2$ clauses positively or negatively; and that in this case, we have $P_{\Omega}(B) = 2^{-k}$, and there are $1 + L k/2$ neighbors of $B$ in the dependency graph. (The factor of $L/2$ here comes from the Lopsided LLL; namely, clauses that intersect on a variable and agree on it, are not counted as dependent for the purposes of the Lopsided LLL.)


Thus, we set $z = \frac{2 \ln \left(\frac{2^{k+1}}{2 + k L}\right)}{2 + k L}$ 
and then we have the bound
\begin{align*}
 -z + P_{\Omega}(B) \exp( \sum_{B' \sim B} z ) &\leq -z + 2^{-k} \exp(z (1 + L k/2)) \\
&= \frac{2 \ln(1 + k L/2) + 2 - k \ln 4}{2 + k L} \\
&= 2^{-k} e \ln (\alpha)/\alpha
\end{align*}

Now, the expected number of resamplings is at most $m z \leq m \log^{O(1)} m / L$. For each resampling, we must scan all the affected clauses to see if they have become falsified, which takes time $k^{O(1)} L \leq L \log^{O(1)} m$. Hence the total expected runtime is $m \log^{O(1)} m$.
\end{proof}

We can also apply this result for \emph{partial} Latin transversals. Although our theorems have been stated in the context of the standard Moser-Tardos algorithm, they only depend on the Witness Tree Lemma. As we have discussed earlier, such results apply in essentially the same way for the permutation-LLL setting described in \cite{lllperm}.
\begin{definition}
Given an $n \times n$ matrix $A$, a \emph{partial Latin transversal} is a selection of $k \leq n$ cells, at most one in each row and column, with the property that there are no two selected cells with the same color.
\end{definition}

Partial Latin transversals have been most studied in the case when $A$ is a Latin square. In \cite{stein:latin}, Stein analyzes the case of partial Latin transversals for arbitrary matrices. Using techniques from that paper, one can show the existence of partial Latin transversals, whose length is a function of $\Delta$, the maximum number of occurrences of any color. This generalizes \cite{erdos-spencer}, which showed that if $\Delta$ is sufficiently small, then a full Latin transversal exists.

\begin{theorem}
\label{stein-version-thm}
Suppose each color appears at most $\Delta = \beta n$ times in the matrix $A$ for $\beta \in [0,1]$. Then one can construct a partial Latin transversal of length at least $n \times \frac{1 - e^{-\beta}}{\beta}$. 
\end{theorem}
\begin{proof}
Suppose that we select a random permutation $\pi$; whenever a color appears more than once in $\pi$, we will remove all but one of those cells from $\pi$ to turn it into a partial Latin transversal.

Suppose that a color appears $d \leq n$ times in the matrix. As shown in \cite{stein:latin}, the probability that $\pi$ meets the color at least once is minimized when all $d$ occurrences of the color are in distinct rows and columns; in this case the probability is (by negative correlation) at least $1-(1 - 1/n)^d$.

Thus, summing over all colors $i$, the total expected number of colors appearing in $\pi$ is at least  $\sum_i 1 - (1 - 1/n)^{d_i}$. By concavity, and using the facts that $d_i \leq \Delta, \sum d_i = n^2$, this is at least $\frac{n^2}{\beta n} (1 - e^{-\beta})$.

Thus, the resulting partial Latin transversal has an expected length of at least $n (\frac{1 - e^{-\beta}}{\beta})$ as we claimed.
\end{proof}

We can improve on Theorem~\ref{stein-version-thm} for $\beta \leq 0.19$ by using the MT-distribution. (Note that for $\beta \leq 0.105$, the LLL constructs a full Latin transversal.)
\begin{theorem}
\label{a1partial-lat-thm}
Suppose each color appears at most $\Delta = \beta n$ times in the matrix $A$, for $\beta \in [0,1/4]$. Then the truncated MT algorithm runs in expected time $O(n)$ and produces a partial Latin transversal whose expected length is at least
$
n \cdot \min\Bigl( 1, \frac{1}{2} + \sqrt[3]{\frac{27}{2048 \beta }} \Bigr).
$
\end{theorem}
\begin{proof}
For every pair of cells $(i,j), (i', j')$ such that $A(i,j) = A(i', j')$, we have a bad- event $\pi(i) = j \wedge \pi(i') = j'$. We apply Theorem~\ref{a1falpha-resamp-thm1}, setting $\mu(B) = \alpha = \frac{1}{n(\Delta-1)} \Bigl( \sqrt[3]{ \frac{n-1}{4 ( \Delta-1) }}
   -1  \Bigr) $ for each such bad-event.  In each independent set of neighbors of a bad-events, for each of the four coordinates $i, j, i', j'$, one may select zero or one bad-events which overlap on that coordinates.

Thus, thew space $\Omega'$ has the property that for each $B$ we have:
{\allowdisplaybreaks
\begin{align*}
P_{\Omega'} (B) &\leq \max(0, \theta(B) -\mu(B)) \\
&\leq \max \bigl( 0, -\alpha + \frac{1}{n(n-1)} (1 + n (\Delta - 1) \alpha)^4 \bigr) \\
& \leq \max\Bigl( 0, \frac{ 1-\frac{3 \sqrt[3]{2}
   \sqrt[3]{n-1}}{8 (\Delta-1)^{1/3}}}{n (\Delta-1)} \Bigr)
\end{align*}
}

Now consider the following experiment: we draw the permutation $\pi$ from the space $\Omega'$. For each bad-event that occurs, we de-activate one of the two cells (chosen arbitrarily). Let $Q$ denote the number of active cells at the end of this process; then $\bE_{\Omega'}[Q] \geq n - \sum_{B} P_{\Omega'}(B)$. 

The total number of bad-events can be computed as follows. First, there are $n^2$ choices for $i,j$. Next, there are $\Delta-1$ choices for $i', j'$. This double-counts the number of bad-events, so in all there are at most $n^2 (\Delta-1)/2$ bad-events.

Thus
\begin{align*}
\bE_{\Omega'}[Q] &\geq n - \frac{n^2 (\Delta-1)}{2} \times \max \Bigl( 0, \frac{ 1-\frac{3 \sqrt[3]{2}
   \sqrt[3]{n-1}}{8 (\Delta-1)^{1/3}}}{n (\Delta-1)} \Bigr) 
\geq n \min\Bigl( 1, \frac{1}{2} + \sqrt[3]{\frac{27}{2048 \beta }} \Bigr)
\end{align*}

\end{proof}

\subsection{A faster parallel (RNC) algorithm}
Suppose we wish to use the parallel MT algorithm to draw from the sample space $\Omega'$ such that:
$$
\forall B \in \mathcal B, P_{\Omega'} (B) \leq \max(0, \theta(B)-\mu(B))
$$

In the symmetric setting (with $e p d = \alpha$), and using the choice of $\mu$ from Corollary~\ref{a1fsym-cor-most}, one can easily verify that the parallel MT algorithm, as described in \cite{moser-tardos}, will terminate after $O( \frac{\log m}{(\alpha-1)^2}) $ rounds whp. (The approach of \cite{hss}, based on two applications of LLL, will give the same result.) The running time of the parallel MT algorithm is dominated by selecting a maximal independent set (MIS) of true bad events (in this case, with the additional property that $Y(B) = 1$). As finding an MIS requires requires $O(\log^2 m)$ parallel time (using Luby's MIS algorithm\cite{lubymis}), the total runtime of parallel MT would be  $O(\frac{\log^3 m}{(\alpha-1)^2})$.

We can improve this running time by only running the parallel MT algorithm for a \emph{constant} number of rounds, using a slightly higher resampling probability than indicated in Theorem~\ref{a1falpha-resamp-thm1}. Unfortunately, we are not able to show a simple condition analogous to the asymmetric LLL for this algorithm to work. Unlike the Moser-Tardos algorithm, which ``converges'' to a good solution, we give an algorithm which ``over-converges'' to the desired solution. It reaches a good distribution faster than Moser-Tardos, but then it moves away from the good distribution. This algorithm seems to require a ``uniformity'' among the bad events, which is by definition true for the Symmetric LLL but seems harder to formalize in general.

We may now define a parallel algorithm corresponding to the Truncated Moser-Tardos Algorithm. It differs from the usual parallel Moser-Tardos algorithm in two key ways. First, we maintain for each bad event $B$ a resampling variable $Y(B)$ which is Bernoulli-$q(B)$, where $q \in [0,1]$ is a parameter to be chosen, and we only resample bad events (including $Y(B)$ itself) when $Y(B) = 1$. Second, instead of running the algorithm until there are no more true bad events, we run it for some fixed number $t$ of iterations. We note that the choice of $q(B)$ is not an ``equilibrium'' value, as in Theorem~\ref{a1falpha-resamp-thm1}; this makes the parallel algorithm more challenging to analyze.

\begin{lemma}
\label{a1ffaster-parallel-cond}
Suppose we are given a family of functions $\sigma_i: \mathcal B \rightarrow [0, \infty)$ for $i = 1, \dots, t+1$ as well as probabilities $q: \mathcal B \rightarrow [0,1]$, satisfying the recurrence for $i = 1, \dots, t$:
\begin{align*}
\sigma_1(B) &\geq q(B) P_{\Omega}(B) \\
\sigma_{i+1}(B) &\geq \sigma_i(B) + q(B) P_{\Omega}(B) \sum_{\substack{\mathcal I \subseteq N(B)\\\text{$\mathcal I$ independent}}} \Bigl[ \prod_{B' \in \mathcal I} \sigma_i(B') - \prod_{B' \in \mathcal I} \sigma_{i-1}(B') \Bigr]
\end{align*}

Then, if the Parallel Truncated Moser-Tardos Algorithm is terminated after $t$ iterations, then each $B$ is true at that point with probability
$$
P(\text{$B$ true after $t$ iterations}) \leq \frac{\sigma_{t+1}(B)}{q(B)} -\sigma_t(B)
$$
\end{lemma}

\begin{proof}
We define $\sigma_0(B) = 0$ for each $B \in \mathcal B$. For each witness tree $\tau$ whose nodes are labeled $B_1, \dots, B_s$, define the weight
$
w(\tau) = \prod_{i=1}^s q(B_i) P_{\Omega}(B_i)
$

Let $T_{i} (B)$ denote the total weight of all witness trees of height $i$ rooted in $B$, and let $T_{\leq i}(B) = \sum_{j \leq i} T_j(B)$. We claim that  $T_{i} (B) \leq \sigma_{i}(B) - \sigma_{i-1}(B)$ for $i = 1, \dots, t$. We shall show this by induction on  $i$. Note that this automatically implies that $T_{\leq i}(B) \leq \sigma_i(B)$ (the sum telescopes).

Suppose $B$ is a tree of height $i$. Let $\mathcal A_1, \mathcal A_2$ denote the sets of neighbors of $B$ whose subtrees have height $i-1$ and $\leq i-2$ respectively. We must have $\mathcal A_1 \neq \emptyset$ in order for $B$ to have height $i$. For a fixed choice of $\mathcal A_1, \mathcal A_2$, the total weight of all such trees is $q(B) P_{\Omega}(B) \prod_{B_1 \in \mathcal A_1}  T_{i-1}(B_1) \prod_{B_2 \in \mathcal A_2} T_{\leq i-2}(B_2)$. Thus, summing over $\mathcal A_1, \mathcal A_2$ we have:
{\allowdisplaybreaks
\begin{align*}
  T_{i} (B) &\leq q(B) P_{\Omega}(B) \smashoperator[l]{\sum_{\substack{\mathcal A_1, \mathcal A_2 \subseteq N(B) \\ A_1 \neq \emptyset, \mathcal A_1 \cap \mathcal A_2 = \emptyset \\ \text{ $\mathcal A_1 \cup \mathcal A_2$ independent}}}} \negthickspace \negthickspace \negthickspace \negthickspace \prod_{B_1 \in \mathcal A_1} T_{i-1} (B_1) \prod_{B_2 \in \mathcal A_2} T_{\leq i-2} (B_2) \\
  &\leq q(B) P_{\Omega}(B) \smashoperator[l]{\sum_{\substack{\mathcal A_1, \mathcal A_2 \subseteq N(B) \\ A_1 \neq \emptyset, \mathcal A_1 \cap \mathcal A_2 = \emptyset \\ \text{ $\mathcal A_1 \cup \mathcal A_2$ independent}}}} \prod_{B_1 \in \mathcal A_1} ( \sigma_{i-1}(B_1) - \sigma_{i-2}(B_1)) \prod_{B_2 \in \mathcal A_2} \sigma_{i-2}(B_2)
\end{align*}
}

In order to evaluate this sum, we first remove the restriction that $\mathcal A_1 \neq \emptyset$, and then we subtract off the terms with $\mathcal A_1 = \emptyset$. In the former case, we would have
{\allowdisplaybreaks
\begin{align*}
&\sum_{\substack{\mathcal A_1, \mathcal A_2 \subseteq N(B) \\  A_1 \cap \mathcal A_2 = \emptyset \\ \text{ $\mathcal A_1 \cup \mathcal A_2$ independent}}} \prod_{B_1 \in \mathcal A_1} ( \sigma_{i-1}(B_1) - \sigma_{i-2}(B_1)) \prod_{B_2 \in \mathcal A_2} \sigma_{i-2}(B_2) \\
&= \sum_{\substack{I \subseteq N(B) \\ \text{$I$ independent}}} \sum_{\substack{\mathcal A_1 \subseteq I \\ \mathcal A_2 = I - \mathcal A_1}} \prod_{B_1 \in \mathcal A_1} ( \sigma_{i-1}(B_1) - \sigma_{i-2}(B_1)) \prod_{B_2 \in \mathcal A_2} \sigma_{i-2}(B_2) \\
&= \sum_{\substack{I \subseteq N(B) \\ \text{$I$ independent}}} \prod_{B' \in I} \Bigl( ( \sigma_{i-1}(B') - \sigma_{i-2}(B') ) + ( \sigma_{i-2}(B') \Bigr) \\
&= \sum_{\substack{I \subseteq N(B) \\ \text{$I$ independent}}} \prod_{B' \in I} \sigma_{i-1}(B')
\end{align*}
}

On the other hand, the contribution from $\mathcal A_1 = \emptyset$ is given by 
\begin{align*}
&\sum_{\substack{\mathcal A_1, \mathcal A_2 \subseteq N(B) \\ A_1 = \emptyset, \mathcal A_1 \cap \mathcal A_2 = \emptyset \\ \text{ $\mathcal A_1 \cup \mathcal A_2$ independent}}} \prod_{B_1 \in \mathcal A_1} ( \sigma_{i-1}(B_1) - \sigma_{i-2}(B_1)) \prod_{B_2 \in \mathcal A_2} \sigma_{i-2}(B_2) \\
&= \sum_{\substack{I \subseteq N(B) \\ \text{$I$ independent}}}  \prod_{B' \in \mathcal I} \sigma_{i-2}(B')
\end{align*}

Putting these together, we have that
\begin{align*}
  T_{i} (B) &\leq q(B) P_{\Omega}(B) \smashoperator[l]{\sum_{\substack{\mathcal I\subseteq N(B) \\ \text{$\mathcal I$ independent}}}} \negthickspace \negthickspace ( \prod_{B' \in \mathcal I} \sigma_{i-1}(B') ) - ( \prod_{B' \in \mathcal I} \sigma_{i-2}(B') ) \\
  &\leq \sigma_{i}(B) - \sigma_{i-1}(B) \qquad \text{(by hypothesis)}
\end{align*}

Now consider the event that bad event $B$ is true after $t$ rounds of the parallel algorithm. We may construct a witness tree for this event; it has height $\leq t+1$. If $Y(B) = 1$ after $t$ rounds, then it must be the case that this tree has height \emph{exactly} $t+1$; for, either $B$ or a neighbor would have been resampled at round $t$. Hence the probability that $B$ remains true after $t$ rounds can be described by either a witness tree of height $t+1$, rooted in $B$; or a witness tree of height $\leq t$, rooted in $(Y(B) = 0) \wedge B$. Furthermore, for every event in the witness tree, other than the root node $B$, we require that $Y(B') = 1$ at the appropriate time. Thus, in total, we have
\begin{align*}
P(\text{$B$ true after $t$ rounds}) &\leq \frac{T_{t+1}(B) + T_{\leq t}(B)  (1- q(B))}{q(B)} \\
&\leq \frac{\sigma_{t+1}(B) - \sigma_t(B) + \sigma_t(B) (1 - q(B))}{q(B)} \\
&=\frac{\sigma_{t+1}(B)}{q(B)} -\sigma_t(B)
\end{align*}
as desired.
\end{proof}

And this specializes to the symmetric setting:
\begin{theorem}
\label{a1fthm:faster-alpha}
Suppose $e p d \leq \alpha$ for $\alpha \in (1, e]$. Then let $\Omega'$ be the distribution induced on the variables after running the Parallel Truncated Moser-Tardos Algorithm for $t$ steps, where $t$ is chosen appropriately as a function of $p, d, \alpha$ and $t = O((\alpha - 1)^{-1})$. In the space $\Omega'$, bad events have probability
$
P_{\Omega'}(B) \leq  \frac{\ln \alpha}{d}. 
$

This can be implemented as a parallel (RNC) algorithm running in $\tilde O( \frac{\log^2 m}{\alpha - 1})$ time. This can also be implemented as a distributed algorithm running in $O( \frac{\log m}{\alpha - 1} )$ rounds (if $p,d,\alpha$ are globally known parameters)

\end{theorem}

\begin{proof}
We note that if $d = 1$, then all the events are completely independent. We can run $t$ rounds of resampling, and each bad-event remains true with probability at most $p^t$. Thus, we need
$t = 1 + \frac{1 + \ln \ln \alpha}{\ln p} \leq O( \frac{\log \log \alpha}{\log (\alpha/e) } ) \leq O( \frac{1}{\alpha - 1})$ rounds of resampling in order to ensure that $p^t \leq \frac{\ln \alpha}{d}$. Henceforth we assume $d \geq 2$.

We next discuss how to select the parameters $t, q$. Let us define
$$
r = \frac{\bigl( \frac{d-1}{d - \ln \alpha} \bigr)^{d-1}}{d}
$$

We claim that $r \geq \frac{\alpha}{e d}$; for $\frac{r e d}{\alpha} = \frac{ e \bigl( \frac{d-1}{d - \ln \alpha} \bigr)^{d-1} }{\alpha}$; this is a decreasing function of $\alpha$, and hence it can be lower-bounded by its value at $\alpha = e$. Thus we have
\begin{align*}
\frac{r e d}{\alpha} &\geq \frac{ e \bigl( \frac{d-1}{d - \ln e} \bigr)^{d-1} }{e} =  \Bigl( \frac{d-1}{d-1} \Bigr)^{d-1} = 1.
\end{align*}

For all $B \in \mathcal B$, define $q(B) = \beta$, for some parameter $\beta$ to be chosen. Define $\sigma_i(B) = \gamma_i(\beta)$ where $\gamma_i(\beta)$ is defined recursively as follows:
\begin{align*}
 \gamma_0(\beta) = 0 \qquad  \gamma_{i+1}(\beta) &= \beta r  (1 + \gamma_i (\beta) )^d
\end{align*}

We first claim that $\gamma_{i+1}(\beta) \geq \gamma_i(\beta)$ for all $i \geq 0$. We show this by induction on $i$. It is clear for $i = 0$. For $i > 0$, we have:
\begin{align*}
\gamma_{i+1}(\beta) &= \beta r (1 + \gamma_i(\beta))^d \\
&\geq \beta r (1 + \gamma_{i-1}(\beta))^d \qquad \text{induction hypothesis} \\
&= \gamma_i(\beta)
\end{align*}

Next, we claim that this definition of $q, \sigma$ satisfies the conditions of Lemma~\ref{a1ffaster-parallel-cond}. For, we have:
\begin{align*}
 &\sigma_i(B) + q(B) P_{\Omega}(B) \smashoperator[l]{\sum_{\substack{\mathcal I \subseteq N(B)\\\text{$\mathcal I$ independent}}}} \negthickspace \negthickspace \prod_{B' \in \mathcal I} \sigma_{i} (B')  - \prod_{B' \in \mathcal I} \sigma_{i-1} (B') \\
 &\qquad = \gamma_i(\beta) + P_{\Omega}(B) \beta {\sum_{\substack{\mathcal I \subseteq N(B)\\\text{$\mathcal I$ independent}}}} \gamma_i(\beta)^{|\mathcal I|} - \gamma_{i-1}(\beta)^{|\mathcal I|} \\
&\qquad \leq \gamma_i(\beta) + P_{\Omega}(B) \beta \Bigl(  (1 + \gamma_i (\beta) )^d - (1 + \gamma_{i-1} (\beta) )^d \Bigr)  \qquad \text{as $|N(B)| \leq d$ and $\gamma_{i} (\beta) \geq \gamma_{i-1} (\beta)$} \\
  &\qquad \leq \gamma_i (\beta) + r \beta \Bigl(  (1 + \gamma_i (\beta) )^d - (1 + \gamma_{i-1} (\beta) )^d \Bigr) \qquad \qquad \text{as $P_{\Omega}(B) \leq p \leq \frac{\alpha}{e d} \leq r$} \\
&\qquad = \gamma_{i+1}(\beta) = \sigma_{i+1}(B)
\end{align*}

Let $z = \frac{1 - \ln \alpha}{d-1} \geq 0$. We claim that for $t$ sufficiently large, there is some $\beta \in [0,1]$ with $\gamma_t(\beta) = z$. We will show this by continuity. Each $\gamma_i(\beta)$ is an increasing function of $\beta$ with $\gamma_i(0) = 0$. Furthermore, we claim that we have for $t \geq 1$:
\begin{equation}
\label{gteqn}
\gamma_t(1)  \geq r \lambda^{t-1} \qquad \text{for $\lambda = r (d-1)(1+1/(d-1))^d$}
\end{equation}

The reason for (\ref{gteqn}) is that for $i \geq 0$ we have $\frac{\gamma_{i+1}(1)}{\gamma_i(1)} = \frac{r (1 + \gamma_i(1))^d}{\gamma_i(1)}$. Now observe that for all $x \geq 0$ we have $\frac{(1 + x)^d}{x} \geq (d-1)(1+1/(d-1))^d$.

Observe that $\lambda = (\frac{d}{d - \ln \alpha})^{d-1} \geq 1$. So, for $t \geq \lceil \max(0,\ln(z/r)) / \ln \lambda \rceil$, we have $\gamma_t(1) \geq z$. Note that $z/r = \bigl( \frac{d - \ln \alpha}{d-1} \bigr)^d d (1 - \ln \alpha)$. Simple calculus shows that this is $O(1)$ for $d \geq 2$. Similarly, simple calculus shows that $\lambda \geq 1 + \Omega(\alpha - 1)$.  So, for $t \geq \Omega( \frac{1}{\alpha - 1} )$ we have that $\gamma_t(1) \geq z$. This implies that there is some $\beta \in [0,1]$ and some choice of $t \leq O( \frac{1}{\alpha - 1} )$ with $\gamma_t(\beta) = z$ exactly.

Now,  Theorem~\ref{a1ffaster-parallel-cond} applies, and so the probability that any $B$ is true after $t$ rounds is at most
\begin{align*}
\frac{\sigma_{t+1}(B)}{q(B)} - \sigma_t(B) &= r (1 + \gamma_t (\beta) )^{d} - \gamma_t (\beta) = r (1 + z)^d - z = \frac{\ln \alpha}{d}
\end{align*}

So far, we have shown by continuity that there is \emph{some} choice of $\beta$, for which the parallel MT algorithm would induce $P_{\Omega'}(B) \leq \frac{\ln \alpha}{d}$. In the distributed setting, where computation is free, we can assume that each node is able to determine this value of $\beta$ to any desired precision. To give a full parallel algorithm, we need to show that it is possible to determine such $\beta$ efficiently. In fact, we only use $\beta$ as a sampling probability; thus, the probability that we need to determine its $i^{\text{th}}$ bit decreases exponentially in $i$. So whp it suffices to compute $O(\log( \frac{m}{\alpha-1} ))$ bits of it.

Recall that $\beta$ is the root of $\gamma_t(\beta) - z$ in the range $\beta \in [0,1]$. We can determine this root via numerical bisection. It requires $O( \log( \frac{m}{\alpha - 1} ))$ rounds of bisection, and each such bisection can be performed in $O( \frac{\log m}{\alpha - 1})$ steps.

\end{proof}

\section{Entropy of the MT-distribution}
\label{a1sec:min}
One of the main themes of this paper has been that the MT-distribution has a high degree of randomness, comparable to the randomness of the original distribution $\Omega$. One more quantitative measure of this is the \emph{R\'{e}nyi entropy} of the MT-distribution.
\begin{definition}[\cite{chor-goldreich:weak-randomness}]
Let $\mathcal V$ be a distribution on a finite set $S$. We define the \emph{R\'{e}nyi entropy} with parameter $\rho$ of $\mathcal V$ to be 
$$
H_{\rho}(\mathcal V) = \frac{1}{1 - \rho} \ln \sum_{v \in S} P_{\mathcal V} (v)^{\rho}
$$
\end{definition}

The entropy of any distribution is at most $\ln |S|$, which is achieved by the uniform distribution, and so $H_{\rho}$ measures how close a distribution is to uniform. The min-entropy $H_{\infty}$ is a special case
$$
H_{\infty}(\mathcal V) = -\ln \max_{v \in S} P(\mathcal V) (v) = \lim_{\rho \rightarrow \infty} H_{\rho}(\mathcal V)
$$
See, e.g., \cite{cohen:two-source,nisan-zuckerman:extractors,vadhan:pseudorandomness} for the centrality of this notion. 

It is possible to use the LLL directly for combinatorial enumeration. Suppose that, when drawing from $\Omega$, the bad-events are avoided with probability with at least $p$; then it follows that the number of solutions is at least $p |S|$. This principle was used in \cite{lu-szekely}, which counted certain types of permutations and matchings in this way.  The entropy can also be used as a tool for enumerative combinatorics; namely, if $\Omega'$ is the distribution at the end of the MT algorithm, we know that the total number of solutions (i.e. combinatorial structures avoiding the bad-events) is at least $\exp(H_{\rho}(\Omega'))$ (for any choice of $\rho$). 

The LLL gives bounds on the number of configurations which are essentially identical to those derived by analyzing the MT distribution. However, the MT distribution has a key advantage, which is that one may efficiently sample from the resulting distribution. The LLL distribution, by contrast, is a conditional distribution. In this sense, one may view the enumerate bounds produced from the MT distribution as being constructive, in a certain sense. Of course, for most applications of the LLL, the number of satisfying assignments is exponentially large, and so it is impossible to give a truly constructive enumerative algorithm for them.

Our main result on the entropy of the MT-distribution is given by:
\begin{theorem}
\label{mt-ent-thm}
Let $\Omega'$ be the MT-distribution; then for $\rho > 1$ we have
$$
H_{\rho}(\Omega') \geq H_{\rho}(\Omega) - \frac{\rho}{\rho - 1} \ln \sum_{\substack{I \subseteq \mathcal B\\\text{$I$ independent}}} \prod_{B \in I} \mu(B)
$$
\end{theorem}
\begin{proof}
Consider some  atomic event $E$ defined by $X_1 = v_1 \wedge \dots \wedge X_n = v_n$. By Proposition~\ref{hss-result}, the probability that $E$ occurs at the end of MT is at most $\theta(E)$. Now observe that $\theta(E) \leq P_{\Omega}(E) \sum_{\substack{I \subseteq \mathcal B \\\text{$I$ independent}}} \prod_{B \in I} \mu(B)$.

Letting $x = \sum_{\substack{I \subseteq \mathcal B\\\text{$I$ independent}}} \prod_{B \in I} \mu(B)$, we thus have:
\begin{align*}
H_{\rho}(\mathcal V) &= \frac{1}{1 - \rho} \ln \sum_{v} P_{\Omega'} (v)^{\rho} \\
&\geq \frac{1}{1 - \rho} \ln \sum_{v} ( x P_{\Omega} (v) )^{\rho} \\
&\geq \frac{\rho}{1 - \rho} \ln x + \frac{1}{1 - \rho}  \sum_{v} P_{\Omega} (v)^{\rho}
\end{align*}

\end{proof}

We can think of the term $\sum_{\substack{I \subseteq \mathcal B\\\text{$I$ independent}}} \prod_{B \in I} \mu(B)$ as a distortion factor between $\Omega$ and $\Omega'$. The following is a crude but simple estimate of this factor:
\begin{proposition}
\label{a1crude-prop}
We have
$$
\ln \sum_{\substack{I \subseteq \mathcal B\\\text{$I$ independent}}} \prod_{B \in I} \mu(B) \leq \sum_{B \in \mathcal B} \mu(B)
$$
\end{proposition}
\begin{proof}
We have
$$\
\sum_{\substack{I \subseteq \mathcal B\\\text{$I$ independent}}} \prod_{B \in I} \mu(B) \leq \prod_{B \in \mathcal B} (1 + \mu(B)) \leq \exp( \sum_{B \in \mathcal B} \mu(B) )
$$
and the claim follows.
\end{proof}

In most applications of the LLL, we keep track of independent sets of bad-events in terms of their variables: namely, for each variable $i$, an independent set $I$ can contain at most one bad-event involving $i$. The following result shows how this variable-based accounting can yield a better estimate for the entropy:

\begin{theorem}
\label{ythm}
For any bad-event $B$, define
$$
y(B) =  (1 + \mu(B))^{\frac{1}{| \text{var}(B)|}} - 1
$$
Then we have 
$$
\sum_{\substack{I \subseteq \mathcal B\\\text{$I$ independent}}} \prod_{B \in I} \mu(B) \leq \prod_{i \in [n]} \Bigl( 1 + \negthickspace \negthickspace \negthickspace \sum_{\substack{B \in \mathcal B \\ \text{$B$ involves variable $i$}}}  \negthickspace \negthickspace \negthickspace y(B) \Bigr)
$$
\end{theorem}
\begin{proof}
We can expand the RHS as a polynomial $Q$ in the values $y(B)$ where $B$ ranges over $\mathcal B$. Given an independent set $I \subseteq \mathcal B$, we say that a monomial in the terms $y$ is \emph{supported} on $I$ if, for each $B$, the exponent of $y(B)$ is positive iff $B \in I$. 

For any set $I$, define $q(I)$ to be the sum of all monomials of $Q$ supported on $I$. Thus, for example if $I = \{B \}$ then $q(I)$ is the sum over all terms in RHS of the form $y(B)^j$, for $j \geq 1$.

Now, observe that if $J, J'$ are distinct subsets of $\mathcal B$, then the monomials supported on $J, J'$ are disjoint. Furthermore, $q(J) \geq 0$ for all $J \subseteq \mathcal B$. Thus
$$
 \prod_{i \in [n]} \Bigl( 1 + \negthickspace \negthickspace \negthickspace \sum_{\substack{B \in \mathcal B \\ \text{$B$ involves variable $i$}}}  \negthickspace \negthickspace \negthickspace y(B) \Bigr) = \sum_{J \subseteq \mathcal B} q(J)
$$

We now claim that for any independent set $I \subseteq \mathcal B$, we have
\begin{equation}
\label{qeqn1}
\prod_{B \in I} \mu(B) = q(I)
\end{equation}

This equation (\ref{qeqn1}) implies that 
\begin{align*}
\sum_{\substack{I \subseteq \mathcal B \\ \text{$I$ independent}}} \prod_{B \in I} \mu(B) &\leq \sum_{\substack{I \subseteq \mathcal B \\ \text{$I$ independent}}} q(I) \\
& \leq \sum_{\substack{J \subseteq \mathcal B}} q(J) \qquad \text{as $q(J) \geq 0$ for all $J \subseteq \mathcal B$} \\
&=  \prod_{i \in [n]} \Bigl( 1 + \negthickspace \negthickspace \negthickspace \sum_{\substack{B \in \mathcal B \\ \text{$B$ involves variable $i$}}}  \negthickspace \negthickspace \negthickspace y(B) \Bigr)
\end{align*}
which is what we are trying to show. So we now move on to prove (\ref{qeqn1}).

For any set $J \subseteq \mathcal B$ (not necessarily independent), we may produce a monomial supported on $J$ by selecting, for each $i = 1, \dots, k$ some set of variables $R_i \subseteq \text{var}(B_i), R_i \neq \emptyset$, and furthermore $R_1, \dots, R_k$ are all disjoint. Thus, for any $J = \{B_1, \dots, B_k \} \subseteq \mathcal B$ we have
\begin{align*}
q(J) &= \sum_{\substack{R_1, \dots, R_k \\  R_1, \dots, R_k \text{disjoint} \\ R_i \subseteq \text{var}(B_i) \\ R_i \neq \emptyset}} y(B_1)^{|R_1|} \dots y(B_k)^{|R_k|}
\end{align*}

Observe that if $I$ is independent, then any such $R_1, \dots, R_k$ are automatically disjoint. Thus, for independent $I = \{ B_1, \dots, B_k \} \subseteq \mathcal B$, we have
$$
q(I) = \negthickspace \negthickspace \sum_{\substack{R_1, \dots, R_k  \\ R_i \subseteq \text{var}(B_i) \\ R_i \neq \emptyset}}  \negthickspace \negthickspace y(B_1)^{|R_1|} \dots y(B_k)^{|R_k|} = \prod_{i=1}^k \sum_{\substack{R \subseteq \text{var}(B_i) \\ R \neq \emptyset}} y(B_i)^{|R|} = \prod_{i=1}^k \bigl( (1 + y(B_i))^{|\text{var}(B_i)|} - 1 \bigr) = \prod_{i=1}^k \mu(B_i)
$$

So, we have shown that for independent $I \subseteq \mathcal B$ we have $q(I) = \prod_{B \in I} \mu(B)$.

\end{proof}

We give an example for independent transversals. Given a graph $G$ with its vertices partitioned into blocks $V = V_1 \sqcup V_2 \sqcup \dots \sqcup V_k$, an \emph{independent transversal} (also known as an \emph{independent system of representatives}) of $G$ is a set $I$ such that $|I \cap V_i| = 1$ for each $i = 1, \dots, k$, and such that $I$ is an independent set of $G$. This. This structure has received significant attention, starting in \cite{bollobas-erdos-szemeredi:isr}. Currently, the best algorithms for producing independent transversals come from the LLL and the MT algorithm; see \cite{bissacot} and \cite{pegden}.

\begin{proposition}
Suppose we have a graph $G$ of maximum degree $\Delta$, with its vertex set partitioned into $k$ blocks containing $ b$ vertices, such that $b \geq 4 \Delta$. Suppose we run the MT algorithm to find an independent transversal, using the natural probability distribution (selecting one vertex independently from each block). Then the MT algorithm terminates and the resulting probability space has min-entropy at least
$$
H_{\infty} (\Omega') \geq k \ln \frac{ 4 b }{2 + b/\Delta - \sqrt{b^2/\Delta^2 - 4 b/\Delta}}
$$
\end{proposition}
\begin{proof}
The min-entropy of $\Omega$ is $-\ln b^{-k} = k \ln b$.

The probability distribution $\Omega$ selects a node from each block uniformly at random. For each edge $f = \langle u, v \rangle \in G$ we have a bad-event that $u,v$ are both selected for the independent transversal. It is any easy exercise to see that the asymmetric LLL criterion is satisfied by setting $\mu(B) = \alpha = \frac{ (b - \sqrt{b^2 - 4 b \Delta} )^2 }{4 b^2 \Delta^2}$ for all $B \in \mathcal B$. Thus, we have $y(B) = ( 1 + \alpha)^{1/2} - 1$.

In this setting, a variable corresponds to a block. There are at most $2 b \Delta$ bad-events involving each block and so we have
\begin{align*}
\prod_{\text{variables $i$}} (1 + \sum_{\text{$B$ involves variable $i$}} y(B)) &\leq \prod_{\text{blocks $i$}} (1 + 2 b \Delta ( (1 + \alpha)^{1/2} - 1) ) \\
&= \Bigl( 1 + 2 b \Delta \bigl( \sqrt{1 + \frac{\left(b-\sqrt{b (b-4
   \Delta)}\right)^2}{4 b^2 \Delta^2}}-1 \bigr) \Bigr)^k
\end{align*}

Now, suppose that $b/\Delta = x$, where $x \geq 4$ is a fixed value; then simple calculus shows that the expression  $1 + 2 b \Delta \bigl (\sqrt{1 + \frac{\left(b-\sqrt{b (b-4
   \Delta)}\right)^2}{4 b^2 \Delta^2}}-1 \big)$ is an increasing function of $\Delta$ which approaches increasingly to $1/2 (x - \sqrt{x (x - 4)})$. Thus, we have that 
$$
1 + 2 b \Delta \bigl (\sqrt{1 + \frac{\left(b-\sqrt{b (b-4 \Delta)}\right)^2}{4 b^2 \Delta^2}}-1 \bigr) \leq \tfrac{1}{2} ( (b/\Delta) - \sqrt{ (b/\Delta) (b/\Delta -4)}).
$$

 By Theorem~\ref{ythm}, this implies that
\begin{align*}
H_{\infty} (\Omega') &\geq k \ln b - k \ln \Bigl( \frac{2 + b/\Delta - \sqrt{b^2/\Delta^2 - 4 b/\Delta}}{4} \Bigr) \\
&= k \ln \frac{ 4 b }{2 + b/\Delta - \sqrt{b^2/\Delta^2 - 4 b/\Delta}}
\end{align*}
\end{proof}

We see that the distortion of $\Omega'$ is relatively mild. When $b = 4 \Delta$, then the min-entropy is $\leq k ( \ln b - \ln 3/2 )$. When $b \gg \Delta$, the min-entropy is (up to first order) $k( \ln b - \frac{\Delta}{2 b} - \frac{7\Delta^2}{8 b^2} - O( (\Delta/b)^{5/2})$. By comparison, the cruder Proposition~\ref{a1crude-prop} would give estimates in these two regimes of, respectively, 
$k(\ln b - 1/2)$ and $k (\ln b - \frac{\Delta}{2 b} - \frac{\Delta^2}{b^2} - O( (\Delta/b)^3)$.

Finally, we give an example for partially satisfying $k$-SAT. This is, to our knowledge, the first result to show that not only is the $k$-SAT problem partially satisfiable, but that it has many partial solutions (indeed, exponentially many solutions).

\begin{proposition}
Suppose we have a $k$-SAT instance with $m$ clauses and $n$ variables, in which each variable participates in up to $L \leq \frac{\alpha 2^{k+1}}{e k} - 2/k$ clauses (either positively or negatively), for $\alpha \in [1,e]$. Then there are at least
$$
\frac{2^n}{ \exp(\frac{\beta (4 + 4 \sqrt{\beta} + \beta)}{k^2}) \text{poly}(m) }
$$
assignments which satisfy at least $m (1 - 2^{-k} e \ln(\alpha)/\alpha) - 1$ clauses, where we define
$$
\beta = 1 - \ln \alpha
$$
\end{proposition}
\begin{proof}
We run the MT algorithm as in Theorem~\ref{a1fk-sat-most} and compute $H_{\rho}$ of the resulting distribution. Using the notation of Theorem~\ref{a1fk-sat-most}, we have $\sum_{B \in \mathcal B} \mu(B) \leq m z$. Observe that, by double-counting $m \leq n L/k$ and so we have $\sum_{B \in \mathcal B} \mu(B) \leq \frac{2 n \beta}{k^2}$.

Next, we compute $H_{\rho}$ of the original distribution. Each variable is Bernoulli with mean $1/2 + x (1/2 - \delta) \leq 1/2 + \frac{\beta}{2 k}$, so we have
$$
H_{\rho}(\Omega) \geq \frac{1}{1 - \rho} \ln \Bigl(  (1/2 + \frac{\beta}{2 k})^{\rho} + (1/2 - \frac{\beta}{2 k})^{\rho} \Bigr)
$$

Hence by Theorem~\ref{mt-ent-thm} we have
\begin{align*}
H_{\rho} (\Omega') &\geq \frac{n}{1 - \rho} \ln \Bigl(  (1/2 + \frac{\beta}{2 k})^{\rho} + (1/2 - \frac{\beta}{2 k})^{\rho} \Bigr) - \frac{\rho}{\rho - 1} \frac{2 n \beta}{k^2}  \\
\end{align*}

We set $\rho = 1 + 2 \beta^{-1/2}$ and use the identity $\ln( (1/2 + w)^{\rho} + (1/2 - w)^{\rho}) \leq (1 - \rho) \ln 2 - 2 (1 - \rho) \rho w^2$ to obtain:
$$
H_{\rho} (\Omega') \geq n \Bigl( \ln 2 - \frac{\beta (4 + 4 \sqrt{\beta} + \beta)}{k^2}  \Bigr)
$$

In the resulting probability distribution, the expected number of failed constraints is $m 2^{-k} e \ln(\alpha)/\alpha$. Hence, by Markov's inequality we fail at most $m 2^{-k} e \ln(\alpha)/\alpha + 1$ constraints with probability at least $\text{poly}(1/m)$. Thus, the entropy of $\Omega'$ conditioned on this event is at least $n( \ln 2 -  \frac{\beta (4 + 4 \sqrt{\beta} + \beta)}{k^2} ) - O(\log m)$. The result follows.

\end{proof}

\section{Acknowledgements}
Thanks to the anonymous journal and conference reviewers for their many  helpful corrections and suggestions.

\end{document}